\documentclass[twocolumn,aps,pra,longbibliography,superscriptaddress,10pt]{revtex4-2}
\usepackage{graphicx}
\usepackage{amsmath}
\usepackage{amssymb}
\usepackage{amsfonts}
\usepackage{amsthm}
\usepackage{mathtools}
\usepackage{color}
\usepackage{hyperref}
\usepackage[T1]{fontenc}
\usepackage[english]{babel}
\usepackage{microtype}
\usepackage{braket}
\usepackage{enumitem}
\usepackage{tikz}
\usepackage{bm}
\usepackage{bbm}
\usepackage{bbold}
\usepackage{physics}
\usepackage{cleveref}
\usepackage{svg}
\usepackage{subcaption}
\usepackage{float}

\usepackage{dsfont}
\usepackage{makecell}
\usepackage[ruled,vlined,linesnumbered]{algorithm2e}
\usepackage{algorithmic}

\usepackage{times}

\hypersetup{colorlinks=true,urlcolor=[rgb]{0,0.6,0.6},citecolor=[rgb]{0,0,0.8},linkcolor=[rgb]{0.5,0,0}}

\usepackage[bbgreekl]{mathbbol}
\DeclareSymbolFontAlphabet{\mathbb}{AMSb}
\DeclareSymbolFontAlphabet{\mathbbl}{bbold}

\newtheorem{theorem}{Theorem}
\newtheorem{definition}{Definition}
\newtheorem{lemma}{Lemma}

\newtheorem{proposition}{Proposition}

\newtheorem{corollary}{Corollary}

\newtheorem{problem}{Problem}
\newtheorem{assumption}{Assumption}

\renewcommand{\epsilon}{\varepsilon}

\allowdisplaybreaks
\newcommand{\vect}[1]{\ensuremath{\mathbf{#1}}}

\newcommand{\E}{\mathbb{E}}
\newcommand{\Prb}{\mathbb{P}}

\DeclareMathOperator{\supp}{supp}

\renewcommand{\Tr}{\operatorname{Tr}}

\renewcommand{\u}{\vect{u}}

\renewcommand{\emptyset}{\varnothing}

\newcommand{\V}{\mathcal V}

\newcommand{\Lop}{\mathcal L}
\newcommand{\Dop}{\mathcal D}

\newcommand{\Id}{\operatorname{Id}}
\newcommand{\one}{\mathbf 1}
\newcommand{\eps}{\varepsilon}
\newcommand{\wt}{\operatorname{wt}}

\begin{document}

\title{Characterizing Arbitrary Lindbladian Dynamics with a Few Pauli Measurements}
\author{Taiqi Zhou}
\email{1155242363@link.cuhk.edu.hk}
\affiliation{Department of Information Engineering, The Chinese University of Hong Kong, Hong Kong, China}
\author{Weiyuan Gong}
\email{wgong@g.harvard.edu}
\affiliation{School of Engineering and Applied Sciences, Harvard University, Allston, MA 02134, USA}

\begin{abstract}
Quantum devices are open systems whose dynamics interleave coherent evolution with dissipation, and benchmarking, error mitigation, and error correction all rest on a faithful model of both. 
Existing characterization protocols either assume prior knowledge of the interaction and noise structure, or demand ancillas, entangled probes, or mid-circuit control, or capture only the Pauli-diagonal part of the noise. 
Here, we present a protocol that reconstructs an arbitrary sparse Markovian generator, including every Hamiltonian together with the jump operator coefficients, using only product Pauli state preparation, single uninterrupted forward evolutions, and product Pauli measurements. 
Given a sparsity budget $M_0$ and a strength bound $\Gamma$ of the Lindbladian, every coefficient is learned to precision $\epsilon$ from $\widetilde{O}(\Gamma^2M_0^2/\epsilon^4)$ experiments and $\widetilde{O}(\Gamma M_0^2/\epsilon^2)$ total evolution time, with both supports identified from data without locality assumptions. 
The protocol runs at a logarithmic number of positive evolution times on a hardware clock lattice and is provably robust to calibrated state-preparation and measurement errors.
\end{abstract}

\maketitle

Quantum devices in labs can be regarded as open systems: their registers evolve under an engineered Hamiltonian while unavoidably exchanging energy and information with an environment.
It is natural to assume that the interaction is Markovian, and the dynamics are thus generated by a Lindbladian~\cite{lindblad1976generators,gorini1976completely,breuer2002theory}, whose coherent part is the Hamiltonian and dissipative part encodes the noise processes~\cite{wolf2008assessing,cubitt2012complexity,rivas2014quantum}.
Characterizing a quantitative model of Lindbladian dynamics serves as a key component across quantum technology~\cite{preskill2018quantum}: it calibrates and certifies analog and digital simulators~\cite{lloyd1996universal,georgescu2014quantum,altman2021quantum,kraft2025bounded}, measures the sensitivity of quantum sensors~\cite{degen2017quantum}, formulate error-correction thresholds and decoder priors~\cite{terhal2015quantum,google2023suppressing,google2025quantum}, and is employed directly by error mitigation—probabilistic error cancellation removes noise by inverting a learned sparse Pauli-Lindblad model~\cite{temme2017error,li2017efficient,cai2023quantum,kim2023evidence,van2023probabilistic}.
In these tasks, the coherent and dissipative contributions are entangled in any measured signal, and mistaking one for the other yields miscalibrated gates, biased mitigation, and misleading benchmarks.
The key question is therefore to characterize the Lindbladian using no more than what the near-term quantum device natively offers: preparing product states, evolving the system, and reading out qubit by qubit.

Despite the practicality requirement, the intriguing structure of real-world quantum dynamics motivates the need for ansatz-free dynamics characterization~\cite{hu2025ansatz}.
While characterizing Lindbladian dynamics without assumptions is possible in principle, it may be prohibitively expensive: process and gate-set tomography reconstruct the full dynamical map at a cost exponential in the number of qubits~\cite{chuang1997prescription,poyatos1997complete,merkel2013self,blume2017demonstration,nielsen2021gate,odonell2016efficient,haah2016sample}, and scalable certification and randomized-measurement toolboxes trade completeness for targeted guarantees~\cite{eisert2020quantum,elben2023randomized}.
Learning many-body dynamics efficiently first became possible by assuming structure. 
Given an interaction ansatz, typically geometric locality, Hamiltonian parameters can be recovered with polynomial resources from local measurements, steady and thermal states, or short-time quenches~\cite{da2011practical,granade2012robust,wiebe2014hamiltonian,wang2015hamiltonian,wang2017experimental,bairey2019learning,qi2019determining,evans2019scalable,anshu2021sample,zubida2021optimal,bakshi2024learning,haah2022optimal,gu2024practical,bakshi2024structure,ma2024learning}, with robustness to imperfections~\cite{francca2024efficient,rouze2024learning,yu2023robust,hangleiter2024robustly,guo2025hamiltonian} and, when interleaved control pulses reshape the evolution, at metrological optimal Heisenberg-limited precision~\cite{huang2023learning,giovannetti2006quantum,giovannetti2011advances,demkowicz2012elusive,zhou2018achieving}.
A line of recent works then removes the ansatz assumption, proposing structure-learning algorithms that discover which interactions are present while estimating their strengths~\cite{hu2025ansatz,zhao2025learning,sinha2025improved,castaneda2025hamiltonian}, reaching Heisenberg scaling without any prior support knowledge~\cite{hu2025ansatz}.
The required quantum resources, however, reach well beyond the capabilities of near-term quantum devices: deep circuits with interleaved probe pulses and extremely fine time resolution~\cite{hu2025ansatz}, engineered inverse evolutions with ancilla assistance~\cite{zhao2025learning}, entangled Bell-basis measurements~\cite{sinha2025improved}, or coherent access to Choi-type encodings~\cite{castaneda2025hamiltonian}.

Under the experimental-friendly restrictions without intermediate control and coherent state preparation or measurements, the protocols are provably weaker~\cite{liu2025optimal} as the total evolution time obeys a standard-quantum-limit floor of order $\epsilon^{-2}$ instead of the Heisenberg scaling $\epsilon^{-1}$~\cite{dutkiewicz2024advantage,chen2025lower,shin2026heisenberg}.
For closed systems, this floor was recently saturated: an ansatz-free Hamiltonian of strength $\Gamma$ can be learned in situ with optimal total evolution time $\Theta((\Gamma/\epsilon^2)\log(\Gamma/\epsilon))$ using only product Pauli preparations and measurements~\cite{zhou2026optimal}.
Ansatz-free characterization for open systems in situ efficiently, however, remains unsolved.
Beyond exponential-cost Lindblad tomography of few-qubit devices~\cite{bouland2003robust,howard2006quantum,samach2022lindblad}, scalable Lindbladian dynamics learning exists almost exclusively for the Pauli-diagonal sector: twirling and randomized compiling enforce an effective Pauli channel~\cite{emerson2005scalable,knill2008randomized,megesan2011scalable,emerson2007symmetrized,wallman2016noise,erhard2019characterizing,helsen2022general} whose error rates, the diagonal of the dissipation term, are efficiently learnable~\cite{flammia2020efficient,flammia2021pauli,harper2021fast,odonnell2026spam} and feed probabilistic error cancellation~\cite{van2023probabilistic}. 
Beyond the diagonal, protocols either assume a local Lindblad ansatz or structural priors~\cite{francca2024efficient,francca2025learning,bairey2020learning,pastori2022characterization,olsacher2025hamiltonian,dobrynin2025compressed,vandenberg2025large} or demand resources far beyond bare devices such as entangled Bell-basis sampling~\cite{sinha2026efficient}, or recursive error-correcting encodings that make QEC itself the learning primitive~\cite{romanov2026learning}
The first ansatz-free in-situ Lindbladian protocol~\cite{ivashkov2026ansatz} is ancilla-free, yet its sample complexity is governed by the conditioning of a linear system that can be exponential in the system size. 
Concurrent structure-learning schemes remain restricted to (quasi-)local generators~\cite{arad2026near,lewis2026learning,mobus2026robust} in in-situ settings, and even detecting dissipation optimally is subtle~\cite{cai2026optimal}.

In this work, we propose a protocol that closes this gap and efficiently reconstructs an arbitrary sparse Lindbladian, including every Hamiltonian and dissipation term coefficient, from product Pauli state preparation, single uninterrupted forward evolutions, and product Pauli measurements available on near-term devices.
Given a sparsity budget $M_0$ and a strength bound $\Gamma$, all coefficients are learned to precision $\eps$ from $\widetilde{O}(\Gamma^2M_0^2/\epsilon^4)$ experiments and $\widetilde{O}(\Gamma M_0^2/\epsilon^2)$ total evolution time, whose $\eps$ and $\Gamma$ dependence saturates the control-free lower bound established for the Hamiltonian special case~\cite{zhou2026optimal}.
Passing from closed to open systems breaks the closed-system toolkit~\cite{zhou2026optimal} in three places: the measurement signals lose the band-limited and trigonometric structure that closed-system protocols exploit, dissipation can damp coherent signatures, and Hamiltonian and dissipative terms collide on shared Pauli labels. 
We resolve these, respectively, with a one-sided polynomial differentiation scheme operating at a logarithmic number of strictly positive times, a positivity (no-jump) argument showing that decoherence can never hide a large Hamiltonian coefficient, and an exact phase-orthogonality identity that separates the two contributions label by label. 
Unlike Refs.~\cite{romanov2026learning,ivashkov2026ansatz}, the guarantees are unconditional, requiring neither structural assumptions nor error-correcting encodings, and they persist under hardware-lattice evolution times and calibrated state-preparation and measurement errors.

\textit{Notations.---}An $n$-qubit state is a density matrix $\rho$.
Pauli strings are labeled by the symplectic space $\V:=\mathbb{F}_2^{2n}$: a label $a=(x|z)$ records the $X$- and $Z$-parts of a Hermitian Pauli string $E_a$, the phases are fixed so that $E_aE_b=\omega(a,b)E_{a+b}$ with $\omega(a,b)\in\{\pm1,\pm i\}$, and $a+b$ denotes the bitwise sum of labels.
Whether two strings commute is recorded by the symplectic bit $[a,b]\in\mathbb{F}_2$ through $E_aE_b=(-1)^{[a,b]}E_bE_a$.
A product Pauli preparation prepares an eigenstate of one chosen Pauli observable on each qubit, and a product Pauli measurement reads out one chosen Pauli observable on each qubit.
A shot whose preparation and measurement bases coincide qubit by qubit is called a \emph{same-Pauli measurement}, and is called a \emph{cross-Pauli measurement} otherwise.
Throughout our protocol, same-Pauli measurements will locate the supports, and cross-Pauli measurements will estimate the coefficient values.

The unknown Lindbladian is expressed in the Pauli-Kossakowski form
\begin{equation}\label{eq:lindbladian}
\begin{split}
\Lop(\rho)&=-i\sum_{s\in\V\setminus\{0\}}h_s[E_s,\rho]\\
&\quad+\sum_{u,v\in\V\setminus\{0\}}A_{uv}\Big(E_u\rho E_v-\tfrac12\{E_vE_u,\rho\}\Big),
\end{split}
\end{equation}
where the real coefficients $h_s$ specify the Hamiltonian $H=\sum_s h_sE_s$, the Hermitian positive semidefinite Kossakowski matrix $A=(A_{uv})$, satisfying $A=A^\dagger\succeq0$, specifies the dissipation, and $\{\cdot,\cdot\}$ represents the anti-commutator. 
The diagonal entry $A_{u,u}$ is the noise rate of the Pauli direction $E_u$, and the possibly complex off-diagonal entries record correlations between different jump directions.
Positive semidefiniteness gives $|A_{uv}|^2\le A_{u,u}A_{v,v}$, so every entry involving a weak diagonal is itself small, a fact used repeatedly below.
We define the supports $S_H:=\{s:h_s\ne0\}$ and $S_D:=\{u:A_{u,u}>0\}$ such that every nonzero entry of $A$ lies inside the block $S_D\times S_D$, and the sparsity is $M:=|S_H|+|S_D|^2$.
The learner is given a budget $M_0\ge M$, the normalization $|h_s|\le1$ and $|A_{uv}|\le1$, and a strength bound $\Gamma\ge\Gamma_*$ on the physical scale
\begin{equation}\label{eq:gammastar}
\Gamma_*:=2\|H\|_\infty+2\!\!\sum_{u,v\in S_D}\!\!|A_{uv}|,
\end{equation}
which bounds every time derivative of the measured responses.
Since $\Gamma_*\le2M_0$, the choice $\Gamma=2M_0$ is always sound.

The access model contains exactly what a near-term quantum device natively offers.
Each experimental shot (i) prepares a product Pauli eigenstate, (ii) evolves the system once under $e^{t\Lop}$ for a chosen time $t\ge0$, (iii) performs a product Pauli measurement, and (iv) postprocesses the record classically.
There are no ancillas, no inverse evolutions, no coherent controls inserted during the evolution, and no mid-circuit operations.
The learning task is to output estimates $(\widehat h,\widehat A)$ such that, with probability at least $1-\delta$, $\max_{s\ne0}|\widehat h_s-h_s|\le\eps$ and $\max_{u,v\ne0}|\widehat A_{uv}-A_{uv}|\le\eps$, while neither $S_H$ nor $S_D$ is known in advance.

\begin{algorithm}[htbp]
\SetAlgoLined
\caption{Control-free Lindbladian reconstruction.}
\label{alg:lindblad}
\KwIn{Parameters $M_0$, $\Gamma\ge\Gamma_*$, $\eps$, and $\delta$}
\KwOut{Estimations $\widehat h=\{\widehat h_s\}$, $\widehat A=\{\widehat A_{uv}\}$}
\tcc{Subroutines}
\begin{enumerate}[label=(\roman*),leftmargin=2.5mm]
\item Fit constants $c,c_0,c_\xi,C_0$ based on the input $(M_0,\Gamma,\eps,\delta)$.
\item A \emph{trace shot} $E_q\to E_{q+s}$ with time $t$ prepares a random $E_q$-eigenstate (sign $r$), evolves $e^{t\Lop}$, measures $E_{q+s}$, and returns $Z=rm$.
\item For a Pauli string label $c$ and $s$, define $(c_i,a_i,b_i)$ cyclic in $(X,Y,Z)$. 
Define Pauli string labels $q_{c,s},o_{c,s}$ as follows: if $s_i=a_i$ or $s_i=c_i$, take $q_{c,s}=a_i$, otherwise take $(q_{c,s})_i=I$. 
Take $o_{c,s}$ such that $q_{c,s}+o_{c,s}=s$
\item Decoder $\texttt{dec}$: $(0,0)\to I$, $(0,1)\to X$, $(1,0)\to Y$,  and $(1,1)\to Z$.
\end{enumerate}
\tcc{The main body}
\lIf{$M_0=0$}{\Return $\widehat h\equiv0,\ \widehat A\equiv0$}
\nl Set $\xi\leftarrow c_\xi\eps^2/M_0$, $\tau_D\leftarrow c_0\xi/\Gamma^2$, $p_\xi\leftarrow\tau_D\xi$, $\tau\leftarrow1/(2\Gamma)$, $R\leftarrow\lceil C_0\log(e{+}\Gamma/\eps)\rceil$\;
\nl Build the distribution $k_R$ on $[0,1]$ by the Legendre kernel, set $B_R\leftarrow\int_0^1|k_R(u)|du$ and distribution $\pi_R\propto|k_R|$, set weight $\kappa(u)\leftarrow\tfrac{B_R}{\tau}\operatorname{sgn}k_R(u)$\;
\tcc{Stage 1: Dissipation structure}
\For{$B\in\{X^n,Y^n,Z^n\}$}{
  Over $N_1\leftarrow\lceil C_0p_\xi^{-1}\log(6M_0/\delta)\rceil$ shots (random $B$-eigenstate signs $x$; evolve $\tau_D$; measure $B\to z$), count each flip pattern $y_i=\one\{z_i\ne x_i\}$\;
  Set $\mathcal Y_B\leftarrow\{y\ne0:\mathrm{count}(y)\ge N_1p_\xi/4\}\cup\{0\}$\;
}
Set $C\leftarrow\{\texttt{dec}(y_X,y_Y):y_X,y_Y,y_X+y_Y\in\mathcal Y_{X^n},\mathcal Y_{Y^n}$,\ $\mathcal Y_{Z^n}\}\setminus\{0\}$, $\Sigma_C\leftarrow\{u{+}v:u,v\in C\}$, $\Delta_C\leftarrow\Sigma_C\setminus\{0\}$\;
\tcc{Stage 2: Hamiltonian structure}
Set $p_*\leftarrow c_0\eps^2/(\Gamma^2R^6)$ and $N_2\leftarrow\lceil C_0p_*^{-1}\log(4M_0/\delta)\rceil$\;
Initialize candidate $X$- and $Z$-projection sets $\widehat\Pi_x,\widehat\Pi_z\leftarrow\varnothing$\;
\For{$(W,\Pi)\in\{(Z,\widehat\Pi_x),(X,\widehat\Pi_z)\}$}{
  Over $N_2$ shots (draw $U\sim\pi_R$; prepare a uniformly random product $W$-eigenstate $\ket{a}_W$; evolve $\tau U$; measure $W\to b$), and add every $a{+}b\ne0$ to $\Pi$\;
}
Set $\widehat U_2\leftarrow\big((\widehat\Pi_x{\cup}\{0\})\times(\widehat\Pi_z{\cup}\{0\})\big)\setminus\{0\}$, $\widehat U_3\leftarrow\widehat U_2\setminus\Delta_C$\;
\tcc{Stage 3: Coefficients}
Set $N_3=\lceil C_0\Gamma^2R^6\eps^{-2}\log(4M_0/\delta)\rceil$\;
\For{$s\in\Sigma_C$}{
  $I_s\leftarrow\{a\in C:a{+}s\in C\}$\;
  \lIf{$s\ne\!0$}{$I_s \leftarrow I_s\cup\{0,s\}$}
  Over $N_3$ shots (draw $U \sim \pi_R$, $q \sim \mathrm{Unif}(\V)$; trace shot $Z$ of $E_q \to E_{q+s}$ at $\tau U$), for all $a \in I_s$: 
  $\widehat R_{a,a+s} \leftarrow \mathrm{mean}[\kappa(U)\,\overline{\gamma_{a,a+s}(q)}\,Z]$, where $\gamma_{a,a+s}(q) = 2^{-n}\Tr(E_{q+s}E_aE_qE_{a+s})$\;
}
Over $N_3$ shots (draw $U \sim \pi_R$, Pauli string with label $c \sim \mathrm{Unif}\{X,Y,Z\}^n$; prepare $a_i$-eigenstates signs $s_i$; evolve $\tau U$; measure $b_i \to m_i$): for each $s \in \widehat U_3$ with an odd number of differences from $c$, $\widehat h_s^{\mathrm{lin}} \leftarrow \mathrm{mean} \big[\tfrac{\kappa(U)Z_s}{2\sigma_c(s)q(s)}\big]$, $Z_s = \prod_{\supp q_{c,s}}  s_i  \prod_{\supp o_{c,s}}  m_i$\;
\tcc{Output}
$\widehat A_{uv} \leftarrow \tfrac12(\widehat R_{u,v}{+}\overline{\widehat R_{v,u}})$ for $u,v \in C$, else $0$\;
$\widehat h_s \leftarrow \operatorname{Re}\tfrac{\widehat R_{0,s}-\widehat R_{s,0}}{2i}$ on $\Delta_C$;\ \ $\widehat h_s \leftarrow \widehat h_s^{\mathrm{lin}}$ on $\widehat U_3$;\ \ else $0$\;
\Return{$\widehat h,\ \widehat A$}\;
\end{algorithm}

\textit{The protocol and performance guarantee.---}Our protocol is summarized in Algorithm~\ref{alg:lindblad}. 
Its performance is guaranteed as follows, as our main result.
\begin{theorem}[Informal, see Theorem 4 in the Supplementary Materials~\cite{supplement}]\label{thm:main}
Given $M_0$, $\Gamma\ge\Gamma_*$, $\eps$, and $\delta$, Algorithm~\ref{alg:lindblad} outputs estimates $(\widehat h,\widehat A)$ satisfying $\max_{s\ne0}|\widehat h_s-h_s|\le\eps$ and $\max_{u,v\ne0}|\widehat A_{uv}-A_{uv}|\le\eps$ with probability at least $1-\delta$, using $\widetilde{O}(\Gamma^2M_0^2/\eps^4)$ experimental shots and $\widetilde{O}(\Gamma M_0^2/\eps^2)$ total evolution time.
\end{theorem}
At the sound scale $\Gamma=2M_0$, the total evolution time reads $\widetilde{O}(M_0^3/\eps^2)$.
Three remarks position the guarantee.
First, the $\eps$ and $\Gamma$ dependence cannot be improved: the per-coefficient evolution time matches the control-free lower bound for the Hamiltonian special case~\cite{zhou2026optimal}, and the $\eps^{-2}$ scaling of the dissipative sector is unavoidable even for adaptive protocols equipped with ancillas~\cite{arad2026near}.
Second, the protocol queries only $O(\log(\Gamma/\eps))$ distinct evolution times, all strictly positive and of order $\Gamma^{-1}$ up to logarithmic factors, and the continuous draws can be snapped onto a hardware clock lattice without changing any scaling (see Section IV. B of the Supplementary Materials~\cite{supplement} for details).
Third, calibrated Pauli-diagonal state-preparation and measurement errors only rescale the measured responses by known factors and are divided out with constant overhead (see Section V of the Supplementary Materials~\cite{supplement} for details).

The algorithm consumes two data types and the constants $c,c_0,c_\xi,C_0$ fitted in subroutine~(i) are absolute constants fixed once by the analysis in the Supplemental Materials~\cite{supplement}.
A trace shot $E_q\to E_{q+s}$ corresponding ot a cross-Pauli measurement or same-Pauli measurement at $s=0$ at time $t$, defined in subroutine~(ii), prepares a uniformly random eigenstate of $E_q$ with recorded sign $r\in\{\pm1\}$, evolves once under a random time drawn below, measures $E_{q+s}$ with outcome $m\in\{\pm1\}$, and returns $Z=rm$, whose average is the two-sided Pauli response
\begin{equation}\label{eq:tracerule}
\E[Z]=2^{-n}\Tr\!\big(E_{q+s}\,e^{t\Lop}(E_q)\big).
\end{equation}
Every coefficient of $\Lop$ therefore lives in the first derivatives of such responses at $t=0$.
A same-Pauli shot instead records the flip pattern $y\in\mathbb{F}_2^n$ between the prepared and measured strings, as collected in line~5.
Line~2 fixes the working scales: the heavy-row threshold $\xi$, the short time $\tau_D$ and its signal scale $p_\xi$ used by Stage~1, the probe window $\tau=1/(2\Gamma)$, and the polynomial degree $R=\lceil C_0\log(e+\Gamma/\eps)\rceil$.
Line~3 builds the signed Legendre kernel $k_R$ on $[0,1]$, its mass $B_R$, the sampling density $\pi_R\propto|k_R|$, and the reweighting $\kappa(u)=(B_R/\tau)\operatorname{sgn}k_R(u)$. 
A random evolution time is then drawn as $t=\tau U$ with $U\sim\pi_R$.

Stage~1 learns the structure of the dissipation terms.
Averaging over the random preparation signs $x$ in line~5 twirls the channel, so each nonzero flip pattern $y$ is sampled with probability given by a diagonal entry of the Pauli-$\chi$ matrix of $e^{\tau_D\Lop}$, whose derivative at $t=0$ is exactly the noise rate $A_{u,u}$~\cite{ivashkov2026ansatz,romanov2026learning}.
At the short time $\tau_D$ of line~2, every heavy row with $A_{u,u}\ge\xi=c_\xi\eps^2/M_0$ therefore produces its pattern with probability at least $p_\xi$, while at most $O(1/p_\xi)$ patterns can pass the counting threshold $N_1p_\xi/4$ of line~6 and enter the list $\mathcal Y_B$, in the spirit of population recovery~\cite{flammia2021pauli}.
Since a row flips the basis $B$ exactly on the sites where its Pauli matrix on the site anticommutes with that of $B$, the three global bases $X^n,Y^n,Z^n$ determine every local Pauli: the pattern pair $(y_X,y_Y)$ fixes the local Paulis through the decoder $\texttt{dec}$ of subroutine~(iv), and their sum must reappear among the $Z^n$ patterns, which is the consistency check in line~8.
The output $C$ contains every heavy dissipative row and satisfies $|C|\le4M_0$. Hence the dissipative structure costs the three measurement settings.

The next step is to learn the structure of the Hamiltonian.
At the first order, the trace response toward a label $s$ receives dissipative contributions only from pairs $u+v=s$ with heavy rows $u,v$, so line~8 also records the collision set $\Sigma_C=\{u+v:u,v\in C\}$ and $\Delta_C=\Sigma_C\setminus\{0\}$. 
Outside $\Delta_C$, the positive semidefiniteness $|A_{uv}|\le(A_{u,u}A_{v,v})^{1/2}$ forces every nonzero $A_{uv}$ of amplitude at most $\eps$ and we can thus output $0$.
Stage~2 then finds the Hamiltonian labels outside $\Delta_C$.
In a same-Pauli shot in the $W=Z$ (respectively $X$) basis in line~12, the displacement $a+b$ of the measured string equals the $x$- (respectively $z$-) projection of the Pauli labels that acted, so nonzero displacements reveal the candidate projections $\widehat\Pi_x$ and $\widehat\Pi_z$ initialized in line~10.
Dissipation could in principle damp this coherent signature.
However, the resolution is the no-jump domination
\begin{equation}\label{eq:nojump}
e^{t\Lop}(\rho)\succeq e^{tK_0}\rho\,e^{tK_0^\dagger},\qquad K_0:=-iH-J/2,
\end{equation}
where $J=\sum_{u,v}A_{uv}E_vE_u$, which lower bounds every displacement probability by the corresponding no-jump amplitude.
Since $J$ is Hermitian, the Hamiltonian enters $K_0$ along the imaginary direction, orthogonal to the real Pauli coefficients of $J$. 
Therefore, any coefficient $|h_s|>\eps$ creates its displacement with probability at least $p_*=c_0\eps^2/(\Gamma^2R^6)$ at a random time $\tau U$, and the $N_2$ shots of line~9 can thus catch it.
The Cartesian product $\widehat U_2$ of line~14, minus the collision set, gives the candidate list $\widehat U_3$ that contains every label with $|h_s|\ge\eps$ outside $\Delta_C$.

The random times drawn from $\pi_R$ are the key reason why no intermediate control is required in Stages~2 and~3.
Without inverse evolutions, no symmetric difference through $t=0$ is available. 
To address this issue, we observe that the scale $\Gamma$ bounds every derivative of a trace response $f$ by $|f^{(\ell)}|\le\Gamma^\ell$, so $f$ is $2^{-R}$-close to a degree-$R$ polynomial on the window $[0,\tau]$.
The signed Legendre kernel of line~3 differentiates degree-$R$ polynomials at the origin exactly, and a single shot at time $\tau U$, reweighted by $\kappa(U)$, returns an estimator with mean $f'(0)$ up to a bias $O(\Gamma R^32^{-R})$ and range $O(\Gamma R^3)$~\cite{trefethen2019approximation}.
The choice of $R$ in line~2 pushes the estimation error below the target and yields $\widetilde{O}(\Gamma^2/\eps^2)$ shots and $\widetilde{O}(\Gamma/\eps^2)$ evolution time per derivative.
The band-limited time sampling of the closed-system protocol~\cite{zhou2026optimal} is unavailable here, because $e^{t\Lop}$ carries no bounded frequency support. 
The endpoint differentiation above is its open-system replacement, and the continuous draws can further be restricted to $R+1$ Chebyshev-Lobatto times (see Section I. F of the Supplementary Materials~\cite{supplement}).

Stage~3 estimates the coefficients with the shot budget $N_3$ of line~15, through two estimators.
On the collision set $\Delta_C$, we expand the generator in the two-sided Pauli basis,
\begin{equation}\label{eq:PL}
\Lop(X)=\sum_{a,b\in\V}R_{a,b}\,E_aXE_b,
\end{equation}
so that $A_{uv}=R_{u,v}$ for nonzero $u,v$. While the Hamiltonian sits on the boundary of the block as $h_s=\operatorname{Re}[(R_{0,s}-R_{s,0})/(2i)]$, the anticommutator contributions cancel in the difference.
For a fixed sum label $s\in\Sigma_C$, the derivative of the trace shot $E_q\to E_{q+s}$ equals $\sum_aR_{a,a+s}\gamma_{a,a+s}(q)$ with the known phases $\gamma_{a,a+s}(q)=2^{-n}\Tr(E_{q+s}E_aE_qE_{a+s})$ defined in line~19.
Over a uniformly random $q$, these phases average orthogonally, so the reweighted mean of line~19 isolates every $\widehat R_{a,a+s}$ in the block from one data set.
Outside the collision set, line~21 runs a randomized linear response: a random configuration $c\in\{X,Y,Z\}^n$ fixes on each site the cyclic triple $(c_i,a_i,b_i)$ of subroutine~(iii), the shot prepares eigenstates along the $a_i$ axes and measures along the $b_i$ axes, and each candidate $s\in\widehat U_3$ splits as $q_{c,s}+o_{c,s}=s$.
Whenever $s$ has an odd number of differences from $c$, the response derivative equals $2\sigma_c(s)h_s$ up to an $O(\eps)$ dissipative bias, where $\sigma_c(s)\in\{\pm1\}$ is a sign computable from $c$ and $s$. 
The parity $Z_s$ multiplies the recorded preparation signs on the support of $q_{c,s}$ with the outcomes on the support of $o_{c,s}$, and dividing by the probability $q(s)\ge1/3$ of the odd-difference event makes $\widehat h_s^{\mathrm{lin}}$ unbiased.
One shot can thus be used to estimate the coefficient of all candidates in $\widehat U_3$ simultaneously.
Lines~22 and~23 assemble the outputs: $\widehat A$ is the Hermitian-symmetrized $\widehat R$ on $C\times C$, and $\widehat h$ takes the boundary formula on $\Delta_C$ and the linear-response value on $\widehat U_3$.

Every coefficient returned as zero in lines~22 and~23 is provably at most $\eps$ with high probability for dissipative entries outside $C\times C$ by positivity, and for Hamiltonian labels outside $\Delta_C\cup\widehat U_3$ by the Stage-2 covering guarantee.
A union bound over the four randomized stages then gives \Cref{thm:main}, with the shot count dominated by Stage~1 and the evolution time dominated by Stage~3, whose $|\Sigma_C|\le16M_0^2$ fixed-sum blocks account for the $M_0^2$ overhead.
Section V of the Supplementary Materials~\cite{supplement} further verifies that snapping all evolution times onto a hardware lattice of spacing $\Delta t=\widetilde{O}(\eps^2/(M_0\Gamma^2))$ and dividing calibrated state-preparation and measurement factors out of every response preserve the scalings of the number of shots and the total evolution time.
\begin{figure}[t!]
\centering
\includegraphics[width=0.49\textwidth]{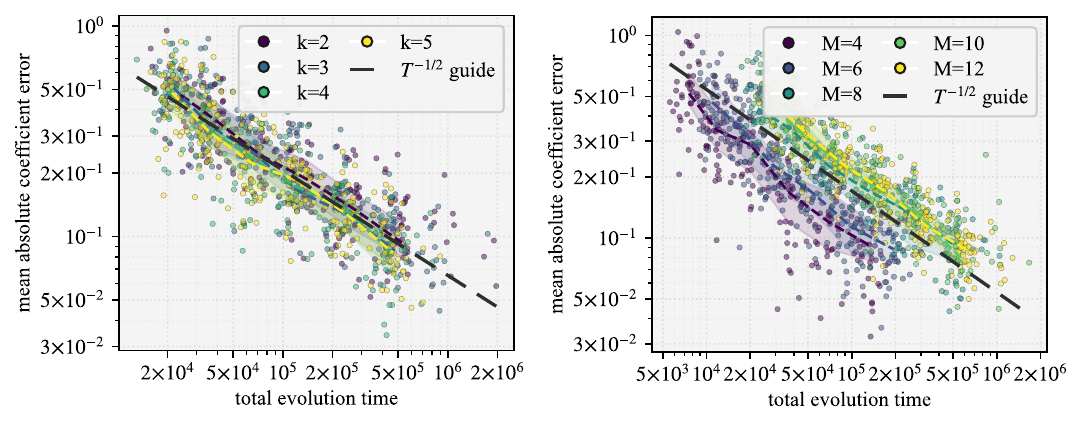}
\caption{(Left) Learning $k$-local Lindbladians. (Right) Learning random sparse Lindbladians.}
\vspace{-2em}
\label{fig:numerics-end-to-end}
\end{figure}

\textit{Numerical simulations.---}We validate the complete Lindbladian reconstruction protocol by numerical simulations in Figure~\ref{fig:numerics-end-to-end}. 
The plotted error is the mean-averaged error
\begin{equation}
\operatorname{MAE}_{S}=\frac{\sum_{s\in S_H}|\widehat h_s-h_s|+\sum_{u,v\in S_D}|\widehat A_{uv}-A_{uv}|}{|S_H|+|S_D|^2},
\end{equation}
where missed coefficients are assigned zero, so both support and estimation errors are included.

For the $k$-local Lindbladian learning case, we set $n=6$ and $k\in\{2,3,4,5\}$, and sample $M_{\mathrm{loc}}=8$ $k$-local Pauli terms and assign them randomly to Hamiltonian and dissipation terms.
For the sparse panel, we set $n=5$ and sample  $M\in\{4,6,8,10,12\}$ Pauli terms and assign them randomly to Hamiltonian and dissipation terms. 
We perform $1200$ and $960$ reconstructions for the $k$-local and the sparse random cases, respectively.
Both cases exhibit the expected SQL scaling $T^{-1/2}$, and the error dependence on $M$ is below the predicted $M^{-3}$ setting $\Gamma=2M_0$.

\textit{Conclusions.---}We have presented a control-free protocol that efficiently reconstructs an arbitrary sparse Lindbladian, including every Hamiltonian and dissipation coefficient, from product Pauli preparations, one uninterrupted forward evolution, and product Pauli measurements, with unconditional guarantees whose precision scaling saturates the control-free limits~\cite{zhou2026optimal,arad2026near} and survives hardware clock lattices and calibrated state-preparation and measurement errors.
Generator-level Lindbladian noise models thus become directly accessible on near-term devices, ready to feed error mitigation beyond Pauli-twirled channels and certified quantum simulation~\cite{van2023probabilistic,kraft2025bounded}.
Two questions remain open.
First, the $M_0^2$ overhead in the total evolution time stems from estimating each fixed-sum block with its own data set, and whether a control-free protocol can approach the sparsity-independent $\widetilde{O}(\Gamma/\eps^2)$ cost of the closed-system case~\cite{zhou2026optimal} is unknown.
Second, extending the present guarantee to time-dependent or weakly non-Markovian dynamics would broaden the reach of in-situ open-system learning.

\textit{Acknowledgements.---}We thank Matthias Caro, Senrui Chen, Sitan Chen, Hsin-Yuan Huang, Zhiding Liang, Nikita Romanov, and Sisi Zhou for the insightful discussions and feedback during the preparation of this work.
A part of this work was carried out while W.G. was visiting the California Institute of Technology.
We acknowledge ChatGPT and Claude for assistance in discussing and refining proof ideas, as well as improving the presentation. 
The authors are solely responsible for the proofs and the results.


\bibliography{control_free_ham_lind_learning_ref}

\clearpage
\newpage
\makeatletter
\setcounter{figure}{0}
\setcounter{equation}{0}
\renewcommand{\thefigure}{S\@arabic\c@figure}
\renewcommand \thetable{S\@arabic\c@table}
\renewcommand \theequation{S\@arabic\c@equation}

\begin{center} 
{\large \bf Supplementary Materials for: Characterizing Arbitrary Lindbladian Dynamics with a Few Pauli Measurements}
\end{center}

\section{Learning Lindbladian from dynamics}
\label{sec:learning-lindbladian-from-dynamics}

\subsection{Basic notations}
\label{subsec:basic-objects}

\noindent
This subsection gives a formal description of the learning problem and collects the basic results needed in the later sections.  
A quantum state of $n$ qubits is described by a density matrix $\rho$, which is a $2^n\times 2^n$ matrix satisfying $\rho\succeq0$ and $\Tr(\rho)=1$.
A pure state $|\psi\rangle$ is represented by the density matrix $\rho=|\psi\rangle\langle\psi|$. 

For an open quantum system, let $\{\Phi_t\}_{t\ge0}$ be a family of continuous-time evolution maps such that $\rho(t)=\Phi_t(\rho(0))$.
The corresponding generator $\Lop$ determines the complete evolution through
\begin{equation}
\Phi_t=e^{t\Lop},\quad
\frac{d}{dt}\rho(t)=\Lop(\rho(t)),\quad
\rho(t)=e^{t\Lop}(\rho(0)).
\label{eq:semigroup-basic}
\end{equation}
In finite dimensions, the generator of a Markovian quantum dynamical semigroup can be written as
\begin{equation}
\Lop(\rho)=-i[H,\rho]+\sum_\alpha\left(L_\alpha\rho L_\alpha^\dagger-\frac{1}{2}\{L_\alpha^\dagger L_\alpha,\rho\}\right),
\label{eq:lindblad-form-basic}
\end{equation}
where $\{A,B\}:=AB+BA$ denotes the anticommutator.
The first term in Eq.~\eqref{eq:lindblad-form-basic} retains the Hamiltonian evolution of a closed system and describes coherent and reversible dynamics.
The second term describes the influence of the environment, and the operators $L_\alpha$ are called jump operators. 

The short-time expansion of the Lindbladian evolution is
\begin{equation}
e^{t\Lop}=\Id+t\Lop+O(t^2).
\label{eq:first-order-channel}
\end{equation}
It follows that the coefficients in the Lindbladian correspond to the first-order changes, at $t=0$, of suitable experimental responses.  
The goal of this work is to use positive-time quantum experiments to recover these generator coefficients that determine the open-system dynamics.

\subsection{Pauli strings and the symplectic form}
To connect this description with the representation and experimental procedures used in this work, we now introduce the Pauli operators.  
Both the unknown Lindbladian and the allowed state-preparation and measurement procedures will be expressed in the Pauli basis.
For one qubit, the basic Pauli matrices are
\begin{equation}
\begin{aligned}
I&=\begin{pmatrix}
1&0\\
0&1
\end{pmatrix},&\qquad
X&=\begin{pmatrix}
0&1\\
1&0
\end{pmatrix},\\
Y&=\begin{pmatrix}
0&-i\\
i&0
\end{pmatrix},&\qquad
Z&=\begin{pmatrix}
1&0\\
0&-1
\end{pmatrix}.
\end{aligned}
\label{eq:pauli-matrices-basic}
\end{equation}
A Pauli measurement means measuring one of $X$, $Y$, or $Z$ on a qubit, with outcome $+1$ or $-1$. 
A product Pauli preparation means preparing a product
state that is an eigenstate of one chosen Pauli observable on each qubit.  
A product Pauli measurement means measuring one chosen single-qubit Pauli observable on each qubit at the end of the evolution.

For a multi-qubit system, we use the symplectic form to record whether two Pauli strings commute or anticommute.
\begin{center}
\begin{tabular}{c|l}
notation & meaning\\ \hline
$E_a$ & Pauli string with binary label $a$\\
$a+b$ & label of $E_aE_b$, ignoring the phase\\
$[a,b]$ & commute-or-anticommute bit\\
$\chi_a(b)=(-1)^{[a,b]}$ & Pauli Walsh character\\
\end{tabular}
\end{center}
Let $\V:=\mathbb F_2^{2n}=\mathbb F_2^n\times\mathbb F_2^n$, a label $a=(x|z)\in\V$ consists of two $n$-bit strings which record the $X$- and $Z$-parts of a fixed Hermitian Pauli string $E_a$.  
We fix the phase of every $E_a$ so Pauli multiplication has the form $E_aE_b=\omega(a,b)E_{a+b}$ where $\omega(a,b)\in\{\pm1,\pm i\}$.
The Pauli strings are orthogonal as $\Tr(E_aE_b)=2^n\one\{a=b\}$.
Multiplying Pauli strings may introduce a phase, but the later arguments only need to know whether reversing their order changes the sign, which is recorded by the symplectic form
\begin{align}
[a,b]=x\cdot z'+z\cdot x' \in\mathbb F_2,\quad a=(x|z),\ b=(x'|z').
\end{align}
In particular, we have $E_aE_b=(-1)^{[a,b]}E_bE_a$, and thus $[a,b]=0$ means commute and $[a,b]=1$ means anticommute. 
The associated Walsh character is $\chi_a(b):=(-1)^{[a,b]}$ which satisfies
\begin{equation}
4^{-n}\sum_{q\in\V}\chi_a(q)\chi_{a'}(q)=\one\{a=a'\}. \label{eq:walsh-orthog}
\end{equation}

\subsection{The Pauli-basis Lindblad model}

We now express the unknown generator in the Pauli basis.  
The real numbers $h_s$ specify the Hamiltonian part, and the Hermitian positive semidefinite matrix $A=(A_{uv})$, satisfying $A=A^\dagger\succeq0$, called the Kossakowski matrix, specifies how different Pauli directions enter the dissipative part. 
Decomposed into the Pauli basis, we have the Pauli-Kossakowski form of the Lindbladian as
\begin{equation}
\begin{aligned}
\Lop(\rho)
&= -i\sum_{s\in\V\setminus\{0\}}h_s[E_s,\rho] \\
&\quad +\sum_{u,v\in\V\setminus\{0\}}A_{uv}
\left(E_u\rho E_v-\frac{1}{2}\{E_vE_u,\rho\}\right).
\end{aligned}
\label{eq:model-L}
\end{equation}
The first sum in Eq.~\eqref{eq:model-L} is the coherent Hamiltonian contribution, and the second sum is the dissipative contribution.  
Its diagonal entries $A_{u,u}$ act as individual Pauli noise rates, while its possibly complex off-diagonal entries describe correlations between Pauli directions.  
For every pair of indices $u,v$, we have according to the fact that $A\succeq 0$
\begin{equation}
|A_{uv}|^2\le A_{u,u}A_{v,v}. 
\label{eq:psd-cauchy}
\end{equation}
Consequently, $A_{u,u}=0$ implies $A_{uv}=A_{vu}=0$ for every $v$.
For every $t\ge0$, the finite-time physical channel is $e^{t\Lop}$.

The support of a coefficient family is the set of labels where its entries are nonzero.  
We define the Hamiltonian support and the dissipative row support by
\begin{align}
S_H:=\{s:h_s\ne 0\},\qquad
S_D:=\{u:A_{u,u} >0\}.
\end{align}
Every nonzero entry $A_{uv}$ lies inside the principal block
$S_D\times S_D$. 
Let $M_H:=|S_H|$, $M_D:=|S_D|$, and define the sparsity as $M:=M_H+M_D^2$.
The known sparsity budget satisfies $M_0\ge M$.
The coefficient normalization is $|h_s|\le 1$ and $|A_{uv}|\le 1$.
For a nonzero instance with $M_0\ge1$ and $0<\eps\le1$, we define
\begin{equation}
\xi:=c_\xi\frac{\eps^2}{M_0}, 
\label{eq:rho-threshold}
\end{equation} 
where $c_\xi>0$ is a sufficiently small universal constant, and set
\begin{equation}
S_{D,\xi}:=\{u\in S_D:A_{u,u}\ge\xi\}. 
\label{eq:heavy-row-set}
\end{equation} 
The reconstruction uses a set $C\supseteq S_{D,\xi}$. Eq.~\eqref{eq:psd-cauchy} controls the entries involving rows in $S_D\setminus C$.

\subsection{Formulation of the learning task}

The algorithm also require a known upper bound $\Gamma\geq\Gamma_*$ of the true physical scale
\begin{equation}
\Gamma_*:=2\|H\|_\infty+2\sum_{u,v\in S_D}|A_{uv}|,\quad H:=\sum_{s\in S_H}h_sE_s. 
\label{eq:true-scale}
\end{equation}
We then have
\begin{align}
\begin{aligned}
\Gamma_*&\le 2\sum_{s\in S_H}|h_s|+2\sum_{u,v\in S_D}|A_{uv}| \\
&\le 2M_H+2M_D^2=2M\le 2M_0.
\end{aligned}
\end{align}
Thus $\Gamma=2M_0$ is always safe if no sharper physical scale is known. 
If $M_0=0$ or $\Gamma=0$, the algorithm returns the zero generator and stops. 
We are then ready to provide the definition of the learning task.

\begin{problem}[Learning sparse Pauli Lindbladian coefficients without control]
Let $\Lop$ be an unknown generator of the form Eq.~\eqref{eq:model-L}. 
The learner is given a sparsity budget $M_0$, a scale bound $\Gamma$, an accuracy parameter $\eps$, and the control-free access model described below.  
The task is to output estimates $(\widehat h,\widehat A)$ such that, with high probability,
\begin{align}
\max_{s\ne0}|\widehat h_s-h_s|\le \eps,\qquad
\max_{u,v\ne0}|\widehat A_{uv}-A_{uv}|\le \eps.
\end{align}
Neither the Hamiltonian support nor the dissipative support is known in advance.
\end{problem}
We call a row $u$ \emph{heavy} when $A_{u,u}\ge\xi$, with $\xi=c_\xi\eps^2/M_0$. 
The algorithm identifies all heavy rows explicitly while entries involving the remaining rows are controlled by $|A_{uv}|^2\le A_{u,u}A_{v,v}$.

We now state exactly what one experimental shot may do.  
Each shot is restricted to the following operations:
\begin{enumerate}[label=(\arabic*)]
\item prepare a product eigenstate of single-qubit Pauli observables;
\item evolve once under $e^{t\Lop}$, with $t\ge0$;
\item measure single-qubit Pauli observables at the end;
\item postprocessing with classical (randomized) computation.
\end{enumerate}
There are no ancillas, no inverse evolution, no coherent controls inserted during the evolution, and no mid-circuit operations.

The protocol extracts two kinds of classical data from these shots:
\begin{enumerate}[label=\Roman*]
\item Cross-Pauli measurements are used to estimate coefficient values. 
\item Same-Pauli measurements are used to locate possible support labels.
\end{enumerate}

\subsection{Cross-Pauli measurements}

Coefficient estimation will repeatedly use the normalized Pauli response $2^{-n}\Tr(Pe^{t\Lop}(Q))$.  
The trace rule below, which is a Lindbladian analogue of~\cite[Proposition 2.3]{zhou2026optimal}, obtains the same quantity by averaging over product eigenstates of $Q$ and multiplying each measurement outcome by the known input eigenvalue.

\begin{proposition}[Product-state data give the normalized Pauli trace response]
Fix Hermitian Pauli strings $P$ and $Q$. 
Sample uniformly from a product eigenbasis of $Q$. 
If the prepared state has $Q$-eigenvalue $r\in\{\pm1\}$, evolve by $e^{t\Lop}$, and we measure $P$ and obtain $m\in\{\pm1\}$, then
\begin{equation}
\E[rm]=2^{-n}\Tr\bigl(Pe^{t\Lop}(Q)\bigr). 
\label{eq:trace-rule}
\end{equation}
\end{proposition}
\noindent Thus the normalized trace response is obtained entirely from ordinary product-state preparations and a classical eigenvalue sign.

\begin{proof}
Let $d=2^n$, and let $\{\psi\}$ be the product eigenbasis from which the preparation is sampled uniformly. 
Write $\rho_\psi:=|\psi\rangle\langle\psi|$ and $Q|\psi\rangle=r(\psi)|\psi\rangle$ for $r(\psi)\in\{\pm1\}$.
For a fixed prepared eigenstate and the channel $e^{t\Lop}$, we have $\E[m\mid \psi]=\Tr\!\left(Pe^{t\Lop}(\rho_\psi)\right)$.
Multiplying by the known preparation eigenvalue and averaging over the $d$ equally likely basis states gives
\begin{align}
\begin{aligned}
\E[rm]&= d^{-1}\sum_\psi r(\psi)\Tr\!\left(Pe^{t\Lop}(\rho_\psi)\right) \\
&= \Tr\Big(Pe^{t\Lop}\!\Big(d^{-1}\sum_\psi r(\psi)\rho_\psi\Big)\Big),
\end{aligned}
\end{align}
where we used linearity of the channel and of the trace. 
The spectral decomposition of $Q$ in this
eigenbasis is
\begin{align}
 Q=\sum_\psi r(\psi)\rho_\psi.
\end{align}
Substituting this identity and using $d=2^n$ yields
\begin{align}
\E[rm]=2^{-n}\Tr\!\left(Pe^{t\Lop}(Q)\right),
\end{align}
which proves the trace identity.
\end{proof}

\subsection{Estimating derivatives from positive times}

The generator appears in the first-order derivative of the channel at $t=0$.  
Since the experiment uses only nonnegative times, we estimate the first-order derivative from weighted measurements at short positive times.

\paragraph{Derivative bounds.}We denote the normalized Pauli response as
\begin{align}
f_{P,Q}(t):=2^{-n}\Tr\bigl(Pe^{t\Lop}(Q)\bigr).
\end{align}
The next lemma shows that the scale $\Gamma$ bounds the derivative.

\begin{lemma}[Upper bound of the derivative for the response function]\label{lem:response-derivative-bound}
For every $\ell\ge 0$, we have the $\ell$-th order derivative $f_{P,Q}^{(\ell)}(t)$ of $f_{P,Q}(t)$ satisfying
\begin{align}
|f_{P,Q}^{(\ell)}(t)|\le \Gamma^\ell.
\end{align}
More generally, if $f_j$ are Pauli responses and $F(t)=\sum_j c_j f_j(t)$ with $\sum_j |c_j|\le 1$, then $|F^{(\ell)}(t)|\le\Gamma^\ell$. 
\end{lemma}

\begin{proof}
We work in the Heisenberg picture. 
We define the adjoint generator as
\begin{align}
\Lop^\dagger(\rho)=i[H,\rho]+
\sum_{u,v\in S_D}A_{uv}\left(E_v\rho E_u-\frac{1}{2}\{E_vE_u,\rho\}\right).
\end{align}
Here $\overline{A_{uv}}=A_{vu}$ is used when taking the Hilbert--Schmidt adjoint and relabeling the indices $u,v$.
For the Hamiltonian term, we note that
\begin{align}
\|i[H,\rho]\|_\infty\le \|H\rho\|_\infty+\|\rho H\|_\infty\le 2\|H\|_\infty\|\rho\|_\infty.
\end{align}
For one dissipative summand, left and right multiplication by Pauli strings are operator-norm
isometries, so
\begin{align}
\begin{aligned}
&\left\|E_v\rho E_u-\frac{1}{2}E_vE_u\rho -\frac{1}{2}\rho E_vE_u\right\|_\infty\le \|E_v\rho E_u\|_\infty\\
&\quad+\frac{1}{2}\|E_vE_u\rho \|_\infty+\frac{1}{2}\|\rho E_vE_u\|_\infty\le 2\|\rho \|_\infty.
\end{aligned}
\end{align}
After multiplying by $|A_{uv}|$ and summing, we obtain
\begin{align}
\begin{aligned}
\|\Lop^\dagger(\rho )\|_\infty&\le \left(2\|H\|_\infty+2\sum_{u,v\in S_D}|A_{uv}|\right)\|\rho \|_\infty \\
&=\Gamma_*\|\rho \|_\infty\le \Gamma\|\rho \|_\infty.
\end{aligned}
\end{align}
Iterating this bound gives
\begin{align}
\| (\Lop^\dagger)^\ell(P)\|_\infty\le \Gamma^\ell\|P\|_\infty=\Gamma^\ell.
\end{align}
Differentiating the response and moving
$\Lop^\ell$ to the adjoint side gives
\begin{align}
\begin{aligned}
f_{P,Q}^{(\ell)}(t)&=2^{-n}\Tr\!\left(P\Lop^\ell e^{t\Lop}(Q)\right) \\
&=2^{-n}\Tr\!\left((\Lop^\dagger)^\ell(P)e^{t\Lop}(Q)\right).
\end{aligned}
\end{align}
For Hermitian inputs, a CPTP map satisfies $\|e^{t\Lop}(X)\|_1\leq\|X\|_1$.
Since $2^{-n}Q$ is Hermitian and has trace norm one, H\"older's inequality for the trace gives
\begin{align}
\begin{aligned}
|f_{P,Q}^{(\ell)}(t)|&=\left|\Tr\!\left((\Lop^\dagger)^\ell(P)e^{t\Lop}(2^{-n}Q)\right)\right| \\
&\le \|(\Lop^\dagger)^\ell(P)\|_\infty\|e^{t\Lop}(2^{-n}Q)\|_1\le \Gamma^\ell,
\end{aligned}
\end{align}
which proves the claimed result.
\end{proof}

\paragraph{One-sided derivative estimator.}We now estimate $f'(0)$ based on Legendre kernel, which is a similar technique employed in Ref.~\cite{zhou2026optimal}.  
On the interval $0\le t\le(2\Gamma)^{-1}$, the derivative bounds make $f$ well approximated by a low-degree polynomial. 
The Legendre kernel below differentiates that polynomial exactly at the left endpoint, and the approximation error is controlled by the degree.

Set $\tau:=\frac{1}{2\Gamma}$ and $u=t/\tau\in[0,1]$.
Let
\begin{align}
\phi_j(u):=\sqrt{2j+1}\,\varphi_j(2u-1),\qquad j=0,1,2,\ldots,
\end{align}
where $\varphi_j$ is the usual Legendre polynomial. 
Define
\begin{align}
\begin{aligned}
&K_R(x,u)&:=\sum_{j=0}^R\phi_j(x)\phi_j(u),\\
&k_R(u)&:=\left.\partial_xK_R(x,u)\right|_{x=0}.
\end{aligned}
\end{align}
Let
\begin{align}
B_R:=\int_0^1 |k_R(u)|\,du,\quad
\pi_R(du):=\frac{|k_R(u)|}{B_R}\,du.
\end{align}
We have the following lemma.

\begin{lemma}[The Legendre kernel returns the derivative at the left endpoint]
For every polynomial $p$ of degree at most $R$,
\begin{align}
\int_0^1 k_R(u)p(u)\,du=p'(0).
\end{align}
Moreover,
\begin{align}
B_R\le \|k_R\|_2\le CR^3,\qquad
\int_0^1k_R(u)\,du=0.
\end{align}
\end{lemma}

\begin{proof}
The functions $\phi_j(u)=\sqrt{2j+1}\,\varphi_j(2u-1)$ are orthonormal in $L^2[0,1]$. 
Hence, we have
\begin{align}
K_R(x,u)=\sum_{j=0}^R\phi_j(x)\phi_j(u)
\end{align}
is the integral kernel of the orthogonal projection onto the space $\mathcal P_R$ of polynomials of degree at most $R$. 
Thus, for every $p\in\mathcal P_R$, one has
\begin{align}
\int_0^1K_R(x,u)p(u)\,du=p(x).
\end{align}
Both sides are polynomials in $x$, so differentiating at $x=0$ is legitimate and gives
\begin{align}
\int_0^1k_R(u)p(u)\,du=\int_0^1\partial_xK_R(x,u)|_{x=0}p(u)\,du=p'(0).
\end{align}
Taking $p\equiv1$ in this identity gives
\begin{align}
\int_0^1k_R(u)\,du=0.
\end{align}
It remains to bound the size of the kernel. 
Since $k_R(u)=\sum_{j=0}^R\phi_j'(0)\phi_j(u)$, the orthonormality gives
\begin{align}
\|k_R\|_2^2=\sum_{j=0}^R|\phi_j'(0)|^2.
\end{align}
Recall that the derivative of the shifted Legendre polynomial is $\phi_j'(0)=2\sqrt{2j+1}\,\varphi_j'(-1)$.
Using the endpoint identity $\varphi_j'(-1)=(-1)^{j+1}\frac{j(j+1)}{2}$, we obtain
\begin{align}
|\phi_j'(0)|=\sqrt{2j+1}\,j(j+1)\le Cj^{5/2}\qquad (j\ge1).
\end{align}
Therefore, we reach
\begin{align}
\|k_R\|_2^2\le C\sum_{j=1}^R j^5\le CR^6,\qquad
\|k_R\|_2\le CR^3.
\end{align}
Finally, because the interval $[0,1]$ has unit measure,
\begin{align}
B_R=\|k_R\|_1\le \|k_R\|_2\le CR^3,
\end{align}
which finishes the proof
\end{proof}

The kernel identity becomes an experimental estimator after randomizing the time intervals.  
Define
\begin{align}
 D_Rf:=\frac{1}{\tau}\int_0^1 k_R(u)f(\tau u)\,du.
\end{align}
To sample this quantity, we draw $U\sim\pi_R$ and output
\begin{align}
Y_R:=\frac{B_R}{\tau}\operatorname{sgn}(k_R(U))Z_{\tau U},\quad \E[Z_t]=f(t)
\end{align}
at $|Z_t|\le 1$.
We then reach the following lemma as the first subroutine of our protocol.

\begin{lemma}[Positive-time data estimate the derivative at zero] \label{thm:one-sided-derivative}
Assume $|f^{(\ell)}(t)|\le \Gamma^\ell$ for all $\ell\ge 0$ and $t\in[0,(2\Gamma)^{-1}]$. 
Then one can estimate $f'(0)$ to additive error $\alpha$ and failure probability $\delta$ using only positive times $t\le(2\Gamma)^{-1}$ and
\begin{align}
O\left(\frac{\Gamma^2R^6}{\alpha^2}\log\frac{1}{\delta}\right)=\widetilde O\left(\frac{\Gamma^2}{\alpha^2}\right)
\end{align}
queries, where $R=C\lceil\log(e+\Gamma/\alpha)\rceil$. 
The total evolution time is $\widetilde O(\Gamma/\alpha^2)$.
\end{lemma}

\begin{proof}
Set $g(u)=f(\tau u)$ on $[0,1]$. 
Since $\tau=1/(2\Gamma)$, the assumed derivative bounds imply
\begin{align}
|g^{(m)}(u)|=\tau^m|f^{(m)}(\tau u)|\le (\tau\Gamma)^m\le 2^{-m}.
\end{align}
Let $T_Rg$ be the Taylor polynomial of $g$ at the endpoint $0$ through degree $R$. 
Then for every $u\in[0,1]$, we have
\begin{align}
\begin{aligned}
|g(u)-T_Rg(u)|&\le \frac{\sup_{0\le v\le1}|g^{(R+1)}(v)|}{(R+1)!}\,u^{R+1} \\
&\le \frac{2^{-(R+1)}}{(R+1)!}\le 2^{-R}.
\end{aligned}
\end{align}
The endpoint kernel is exact on $T_Rg$, so
\begin{align}
\int_0^1k_R(u)T_Rg(u)\,du=(T_Rg)'(0)=g'(0).
\end{align}
Consequently
\begin{align}
\begin{aligned}
&\left|\int_0^1k_R(u)g(u)\,du-g'(0)\right| \\
&\quad=\left|\int_0^1k_R(u)(g(u)-T_Rg(u))\,du\right| \\
&\quad\le B_R2^{-R}\le CR^3 2^{-R}.
\end{aligned}
\end{align}
Since $D_Rf=\tau^{-1}\int_0^1k_R(u)g(u)\,du$ and $g'(0)=\tau f'(0)$, we have
\begin{align}
|D_Rf-f'(0)|\le \tau^{-1}CR^3 2^{-R}=C\Gamma R^3 2^{-R}.
\end{align}
Choosing $R=C_0\lceil\log(e+\Gamma/\alpha)\rceil$ with $C_0$ sufficiently large makes this bias at most $\alpha/2$.

Now draw $U\sim\pi_R$ and, conditional on $U$, run one trace-rule experiment producing $Z_{\tau U}$ with $\E[Z_{\tau U}\mid U]=f(\tau U)$ and $|Z_{\tau U}|\le1$. 
The weighted sample
\begin{align}
Y_R=\frac{B_R}{\tau}\operatorname{sgn}(k_R(U))Z_{\tau U}
\end{align}
satisfies $|Y_R|\le \frac{B_R}{\tau}\le C\Gamma R^3$ and
\begin{align}
\begin{aligned}
\E[Y_R]&=\frac{B_R}{\tau}\int_0^1\operatorname{sgn}(k_R(u))f(\tau u)\frac{|k_R(u)|}{B_R}\,du \\
&=\frac{1}{\tau}\int_0^1k_R(u)f(\tau u)\,du=D_Rf.
\end{aligned}
\end{align}
Hoeffding's inequality applied to the average of $N$ independent copies gives sampling error at most $\alpha/2$ with probability at least $1-\delta$ provided
\begin{align}
N\ge C\frac{\Gamma^2R^6}{\alpha^2}\log\frac{1}{\delta}.
\end{align}
Combining this sampling error with the deterministic bias proves the stated additive accuracy. 
If $f$ is complex-valued, the same argument is applied separately to the real and imaginary parts, and the failure probabilities are split by a constant factor. 
Each implemented evolution time is at most $\tau=O(1/\Gamma)$, so the total evolution time is $N\tau=\widetilde O(\Gamma/\alpha^2)$.
\end{proof}

The structure-learning stage uses a vector of transition amplitudes rather than one scalar response.  
The same kernel argument gives the following vector-valued statement.

\begin{lemma}[A large initial derivative produces a detectable positive-time signal]\label{lem:vector-observability}
Let $F:[0,(2\Gamma)^{-1}]\to\mathcal H$ be a finite-dimensional Hilbert-space-valued function with $\|F^{(\ell)}(t)\|\le \Gamma^\ell$. 
If $\|F'(0)\|\ge \eta$, then with $R=C\lceil\log(e+\Gamma/\eta)\rceil$, one have
\begin{align}
\E_{U\sim\pi_R}\|F(\tau U)\|^2\ge \frac{c\eta^2}{\Gamma^2R^6}.
\end{align}
\end{lemma}

\begin{proof}
Let $G(u)=F(\tau u)$. 
The Taylor-remainder argument in the proof of the scalar derivative estimator uses only the triangle inequality in the target space. 
The same proof therefore applies in the Hilbert space $\mathcal H$. 
Thus, for the same choice of $R$, we have
\begin{align}
\left\|\frac{1}{\tau}\int_0^1 k_R(u)F(\tau u)\,du-F'(0)\right\|\le C\Gamma R^3 2^{-R}\le \frac{\eta}{2}.
\end{align}
Since $\|F'(0)\|\ge\eta$, the reverse triangle inequality gives
\begin{align}
\left\|\frac{1}{\tau}\int_0^1 k_R(u)F(\tau u)\,du\right\|\ge \frac{\eta}{2}.
\end{align}
Write the same vector as an expectation under $\pi_R$:
\begin{align}
\frac{1}{\tau}\int_0^1 k_R(u)F(\tau u)\,du=\frac{B_R}{\tau}\E_{U\sim\pi_R}\left[\operatorname{sgn}(k_R(U))F(\tau U)\right].
\end{align}
By Jensen or Cauchy-Schwarz,
\begin{align}
\begin{aligned}
\left\|\frac{1}{\tau}\int_0^1 k_R(u)F(\tau u)\,du\right\|^2
&\le \left(\frac{B_R}{\tau}\right)^2\E_{U\sim\pi_R}\|F(\tau U)\|^2.
\end{aligned}
\end{align}
Using $B_R\le CR^3$ and $\tau^{-1}=2\Gamma$, we have $(\tfrac{B_R}{\tau})^2\le C\Gamma^2R^6$.
Combining this with the lower bound by $\eta/2$ gives
\begin{align}
\E_{U\sim\pi_R}\|F(\tau U)\|^2\ge \frac{(\eta/2)^2}{C\Gamma^2R^6}\ge \frac{c\eta^2}{\Gamma^2R^6},
\end{align}
after absorbing universal constants into $c$.
\end{proof}

\section{Lindbladian structure learning}
\label{sec:lindbladian-structure-learning}

Before estimating coefficient values, the algorithm first reduces the $4^n-1$ possible Pauli labels to small candidate sets. 
Same-Pauli measurements find a set $C$ containing every heavy dissipative row. 
Pairwise sums of labels in $C$ form the collision set $\Delta_C=\{u+v:u,v\in C\}\setminus\{0\}$, where dissipative and Hamiltonian first-order terms may share the same label.  Same-Pauli measurements then find the Hamiltonian candidates outside this set.

\subsection{Same-Pauli measurements and diagonal Pauli probabilities}

Write an $n$-qubit channel $\Phi$ in its
Pauli-$\chi$ representation as
\begin{align}
\Phi(\rho)=\sum_{a,b\in\V}\chi_{a,b}(\Phi)E_a\rho E_b.
\end{align}
The matrix $\chi(\Phi)$ is a process matrix in the Pauli basis.  
Its diagonal entries become Pauli-error probabilities after Pauli twirling.  
In particular, $\chi_{u,u}(\Phi)$ is the probability of the error $\rho\mapsto E_u\rho E_u$ in the twirled channel.  
When $\Phi=e^{t\Lop}$, the derivative of this probability at $t=0$ is the rate at which that Pauli error is created. 
The original channel itself need not be Pauli diagonal.

A Pauli twirl is the mathematical average obtained by conjugating the input and output by all Pauli strings.  
We use it only to describe the measurement statistics. 
The experiment does not physically apply these conjugations. 
Define
\begin{align}
\Phi^\circ(X):=4^{-n}\sum_{r\in\V}E_r\,\Phi(E_rXE_r)\,E_r.
\end{align}
If $\Phi(X)=\sum_{a,b}\chi_{a,b}(\Phi)E_aXE_b$, then the twirl multiplies the term $E_aXE_b$ by the average of $(-1)^{[r,a+b]}$ over $r\in\V$. 
Walsh orthogonality removes all terms with $a\ne b$, and gives
\begin{align}
\Phi^\circ(X)=\sum_{u\in\V}\chi_{u,u}(\Phi)E_uXE_u.
\end{align}
The diagonal $\chi$-entries are therefore the Pauli-error probabilities of the Pauli-twirled map. 
Since $\Phi^\circ$ is a trace-preserving Pauli channel,
\begin{align}
\chi_{u,u}(\Phi)\ge 0,\qquad
\sum_{u\in\V}\chi_{u,u}(\Phi)=1.
\end{align}
We write this classical probability vector as $p_u=\chi_{u,u}(\Phi)$.  
We will use population recovery to recover the entries of this probability vector~\cite{flammia2021pauli}.

A same-Pauli measurement uses the same local Pauli basis before and after the evolution, the goal of which is to record only which local eigenvalue signs changed.

\begin{definition}[How a same-Pauli measurement records local sign flips]\label{def:pauli-pr-probe}
Fix a basis string $a=(a_1,\ldots,a_n)\in\{X,Y,Z\}^n$.
For each sign string $x\in\{\pm1\}^n$, let $\Pi_x^a$ be the rank-one product projector onto the simultaneous eigenspace with local eigenvalue $x_i$ for $a_i$. 
The same-Pauli measurement for $\Phi$ is defined as sampling $x\sim{\rm Unif}(\{\pm1\}^n)$, preparing the product state $\Pi_x^a$, applying $\Phi$, and measuring every qubit again in the same local bases $a_i$. 
If the measured sign string is $z\in\{\pm1\}^n$, record the binary string $U\in\mathbb F_2^n$ defined by
\begin{align}
U_i =\begin{cases}
0, & z_i=x_i,\\
1, & z_i=-x_i.
\end{cases}
\end{align}
For a Pauli label $u\in\V$, define its flip pattern in the same-Pauli measurement relative to $a$ by
\begin{align}
\alpha_a(u)_i:=\one\{E_{u_i}a_i=-a_iE_{u_i}\},\qquad i=1,\ldots,n.
\end{align}
Equivalently, $\alpha_a(u)_i=1$ exactly when the local Pauli letter $u_i$ anticommutes with the measured local basis $a_i$.
\end{definition}

The uniformly random input sign is essential: after averaging over it, all off-diagonal Pauli-$\chi$ terms cancel as a standard subroutine of Pauli twirling~\cite{emerson2007symmetrized,wallman2016noise,van2023probabilistic} as in the following lemma (see, e.g., the above-mentioned references for the proof). 
The recorded flip string is therefore a coarse-grained sample from the diagonal Pauli-error distribution.

\begin{lemma}[Same-Pauli measurements depend on diagonal Pauli-$\chi$ probabilities]
\label{lem:product-pauli-diagonal-statistics}
Let $\Phi(\rho)=\sum_{r,s\in\V}\chi_{r,s}(\Phi)E_r\rho E_s$ be the Pauli-$\chi$ expansion of an arbitrary CPTP channel. In a same-Pauli measurement with basis string $a$, the recorded flip string satisfies
\begin{align}
\Prb_\Phi[U=y\mid a]=\sum_{u:\,\alpha_a(u)=y}\chi_{u,u}(\Phi).
\end{align}
Then the same-Pauli measurement data of $\Phi$ have the same distribution as the data of the Pauli channel
\begin{align}
\Phi^\circ(\rho)=\sum_{u\in\V}\chi_{u,u}(\Phi)E_u\rho E_u.
\end{align}
\end{lemma}

The next lemma, which is Theorem 2 of Ref.~\cite{flammia2021pauli}, serves as the second component of our protocol, and we include the construction because the later row-support proof uses its uniform error bound over all Pauli labels.

\begin{lemma}[Product-Pauli data recover all large Pauli-error probabilities, Theorem 2 of Ref.~\cite{flammia2021pauli}]\label{thm:population-recovery}
For any $n$-qubit channel $\Phi$, there is an ancilla-free procedure with product-Pauli input and output, and the following output guarantee. 
Let $p_u:=\chi_{u,u}(\Phi)$ for $u\in\V$.
The procedure returns a sparse hypothesis list
\begin{align}
\widehat p=\{(u,\widehat p_u):u\in\widehat{\mathcal S}\},
\qquad |\widehat{\mathcal S}|\le C_p/\epsilon,
\end{align}
and uses the convention $\widehat p_u:=0\qquad (u\notin\widehat{\mathcal S})$.
With probability at least $1-\delta$, we have $\sup_{u\in\V}|\widehat p_u-p_u|\le \epsilon$.
The number of uses of $\Phi$ is
\begin{align}
m=O\!\left(\epsilon^{-2}\log\frac{n}{\epsilon\delta}\right),
\end{align}
and the classical postprocessing time is polynomial in $n$ and $1/\epsilon$. 
More precisely, the standard sparse-tree implementation has postprocessing time $O(mn/\epsilon)$ up to absolute constants.
\end{lemma}

\subsection{Diagonal slopes reveal dissipative rows}

The diagonal Pauli-$\chi$ probability distribution is useful because it equals the corresponding diagonal entry of the Kossakowski matrix, i.e., $\chi'_{u,u}(0)=A_{u,u}$~\cite{ivashkov2026ansatz,romanov2026learning}. 
The row-finding algorithm estimates this distribution from positive times, so it needs uniform derivative bounds for $\chi_{u,u}(t)$.  
The next two lemmas are variants of Lemma C.1 of Ref.~\cite{ivashkov2026ansatz} with the physical scale $\Gamma$.

\begin{lemma}[The physical scale bounds the Hilbert-Schmidt norm of the generator, Lemma C.1 of Ref.~\cite{ivashkov2026ansatz}]\label{lem:hs-generator-bound}
With $\Gamma_*$ defined in Eq.~\eqref{eq:true-scale}, we have
\begin{align}
\|\Lop\|_{2\to2}\le \Gamma_*\le \Gamma.
\end{align}
Here, $\|\cdot\|_{2\to2}$ is the induced Hilbert-Schmidt norm on operators.
\end{lemma}

The short-time probability bound used below must remain valid after any allowed preparation and measurement.  
For this purpose we use the diamond norm, which controls the action of a channel even when an auxiliary system is present.  
The same termwise estimate gives
\begin{align}
\|\Lop\|_\diamond\le 2\|H\|_\infty+2\sum_{u,v\in S_D}|A_{uv}|=\Gamma_*\le\Gamma.
\end{align}
Indeed, the diamond norm of the two-sided multiplication map $X\mapsto AXB$ is at most $\|A\|_\infty\|B\|_\infty$.  
Hence, Pauli left and right multiplication have diamond norm one, the commutator part has diamond norm at most $2\|H\|_\infty$, and each dissipative basis map has diamond norm at most two.

\begin{lemma}[The physical scale bounds every derivative of a Pauli-$\chi$ coefficient, Lemma C.1 of Ref.~\cite{ivashkov2026ansatz}]
For every $k\ge 0$, we have
\begin{align}
\left|\frac{d^k}{dt^k}\chi_{a,b}(t)\right|\le \Gamma^k.
\end{align}
\end{lemma}

\subsection{Finding heavy dissipative rows from three Pauli bases}

We now convert the diagonal distribution computation subroutine into an explicit row-finding procedure. 
A local Pauli matrix is determined by whether it commutes or anticommutes with $X$ and with $Y$.  
Therefore, the global $X^n$ and $Y^n$ flip patterns identify the label coordinate by coordinate, while the $Z^n$ pattern checks that the two decoded patterns are consistent.

For $B\in\{X^n,Y^n,Z^n\}$ and a Pauli label $u$, define a vector $\alpha_B(u)\in\mathbb{F}_2^n$
\begin{align}
\alpha_B(u)_i:=\one\{E_{u_i}B_i=-B_iE_{u_i}\}\in\mathbb F_2.
\end{align}
For $y\in\mathbb F_2^n$, let
\begin{align}
\lambda_B(y):= \sum_{u\in S_D:\,\alpha_B(u)=y}A_{u,u}.
\end{align}
\Cref{lem:product-pauli-diagonal-statistics} gives
\begin{align}
\Prb[U_B=y]=\sum_{u:\alpha_B(u)=y}\chi_{u,u}(t).
\end{align} 
For $y\ne0$, Taylor expansion at $t=0$ yields
\begin{equation}
\begin{split}
&\Prb[U_B=y]=t\lambda_B(y)+R_B(y),\\
&\sum_{y\ne0}|R_B(y)|\le C_1\Gamma^2t^2. 
\end{split}
\label{eq:bucket-short-time}
\end{equation} 
The first-order derivative is thus $\lambda_B(y)$, and the second equality follows from the fact that, by the diamond-norm estimate above and submultiplicativity, for $\Gamma t\le1$, we have
\begin{align}
\|e^{t\Lop}-\Id-t\Lop\|_\diamond\le e^{\Gamma t}-1-\Gamma t\le C_1\Gamma^2t^2,
\end{align}
and applying the preparation and measurement maps and taking the classical $\ell_1$ norm gives Eq.~\eqref{eq:bucket-short-time}.

A row with $A_{u,u}\ge\xi$ creates a first-order flip probability in each basis where its pattern is nonzero.  
Repeating the three same-Pauli measurements therefore finds every such row. 
The procedure is allowed to return extra labels, but the list remains limited-size.
Quantitatively, we have the following key lemma.

\begin{lemma}[Three global Pauli bases find every heavy dissipative row] \label{thm:diss-row}
Let $0<\eps\le1$ and let $\xi$ be defined by Eq.~\eqref{eq:rho-threshold}. 
A control-free product-Pauli procedure outputs $C\subseteq\V\setminus\{0\}$ such that, with probability at least $1-\delta$, we have $S_{D,\xi}\subseteq C$ and $|C|\le 4M_0$.
Moreover, it uses
\begin{align}
N_D=\widetilde O\!\left(\frac{\Gamma^2}{\xi^2}\right)=\widetilde O\!\left(\frac{\Gamma^2M_0^2}{\eps^4}\right)
\end{align}
shots and total evolution time
\begin{align}
T_D=\widetilde O\!\left(\frac{1}{\xi}\right)=\widetilde O\!\left(\frac{M_0}{\eps^2}\right).
\end{align}
\end{lemma}

\begin{proof}
If $\Gamma<2\xi$, then
\begin{align}
A_{u,u}\le \sum_{v,w\in S_D}|A_{v,w}|\le \Gamma/2 <\xi
\end{align}
for every $u$, so $S_{D,\xi}=\emptyset$ and the procedure returns $C=\emptyset$. 
If $\Gamma\ge2\xi$, we set
\begin{align}
\tau_D:=c_0\frac{\xi}{\Gamma^2},\qquad
p_\xi:=\tau_D\xi=c_0\frac{\xi^2}{\Gamma^2},
\end{align}
with $c_0>0$ sufficiently small. 
In each basis $B\in\{X^n,Y^n,Z^n\}$, we perform
\begin{align}
N_B=Cp_\xi^{-1}\log\frac{6(M_0)}{\delta}
\end{align}
independent shots at time $\tau_D$ assuming $M_0\geq1$. 
Assuming $C_B(y)$ is the number of times the pattern $y$ appears, we set
\begin{align}
\mathcal Y_B:=\{y:C_B(y)\ge N_Bp_\xi/4\}\cup\{0\}.
\end{align}
For $u\in S_{D,\xi}$ with $y=\alpha_B(u)\ne0$, Eq.~\eqref{eq:bucket-short-time} and $\lambda_B(y)\ge A_{u,u}\ge\xi$ imply
\begin{align}
\Prb[U_B=y]\ge \tau_D\xi-C_1\Gamma^2\tau_D^2\ge\frac{15}{16}p_\xi.
\end{align}
A Chernoff bound, followed by a union bound over the three bases and at most $M_D$ heavy rows, shows that every heavy-row pattern belongs to the corresponding list $\mathcal Y_B$.
 
Let $\mathcal S_B:=\{\alpha_B(u):u\in S_D\}\cup\{0\}$.
Then $|\mathcal S_B|\le M_D+1$. 
For $y\notin\mathcal S_B$, the first-order mass vanishes, and
\eqref{eq:bucket-short-time} gives
\begin{align}
\Prb[U_B\notin\mathcal S_B]\le C_1\Gamma^2 \tau _D^2\le p_\xi/16.
\end{align}
The number of samples outside $\mathcal S_B$ is then smaller than $N_Bp_\xi/4$ with probability at least $1-\delta/6$, after increasing $C$. 
Conditioned on this event, no pattern outside $\mathcal S_B$ is retained, so $|\mathcal Y_B|\le M_D+1$.

For every single-qubit Pauli matrix, we have
\begin{align}
(\alpha_X,\alpha_Y,\alpha_Z)\in\{(0,0,0),(0,1,1),(1,0,1),(1,1,0)\},
\end{align}
and hence
\begin{align}
\alpha_{Z^n}(u)=\alpha_{X^n}(u)+\alpha_{Y^n}(u).
\end{align}
The pair $(\alpha_{X^n}(u),\alpha_{Y^n}(u))$ then determines each local Pauli letter through
\begin{align}
(0,0)\mapsto I,\ (0,1)\mapsto X,\ (1,0)\mapsto Y,\ (1,1)\mapsto Z.
\end{align} 
Form $C$ from all pairs $(y_X,y_Y)\in\mathcal Y_{X^n}\times\mathcal Y_{Y^n}$ satisfying $y_X+y_Y\in\mathcal Y_{Z^n}$, and discard the identity label. 
The preceding inclusion places every $u\in S_{D,\xi}$ in $C$. 
If $M_D=0$, the decoded set is empty. 
If $M_D\ge1$, then we have
\begin{align}
|C|\le(M_D+1)^2\le4M_D^2\le4M_0.
\end{align} 
Finally, $N_D=3N_B$ and $T_D=3N_B\tau_D$, which give the stated complexity.
\end{proof}

\subsection{The collision set \texorpdfstring{$\Delta_C$}{Delta C}}

An off-diagonal coefficient $A_{uv}$ contributes one-sided terms with Pauli label $u+v$.  
A Hamiltonian coefficient $h_s$ also contributes at the one-sided label $s$. 
Thus a Hamiltonian and a dissipative term can share a label only when that label is a sum of two dissipative rows.  Since the true rows are not known exactly, we form the corresponding candidate sum set from $C$.

Conditioned on the success case of \Cref{thm:diss-row}, we define
\begin{equation}
\Sigma_C:=\{u+v:u,v\in C\},\qquad
\Delta_C:=\Sigma_C\setminus\{0\}. 
\label{eq:Delta-C}
\end{equation} 
For $s\ne0$, set
\begin{align}
\beta_C(s):=\sum_{\substack{u+v=s\\u,v\in S_D}}|A_{uv}|.
\end{align}
We have the following lemma.

\begin{lemma}[Outside the candidate sum set, the total dissipative contribution is small]\label{lem:weak-tail}
Assume $S_{D,\xi}\subseteq C$. Then for any $u\in S_D\setminus C$, we have $A_{u,u}<\xi$ and $|A_{uv}|\le\sqrt\xi$ for any $v\in S_D$.
For every $s\notin\Delta_C$,
\begin{equation}
\beta_C(s)\le\sqrt{M_0\xi}=\sqrt{c_\xi}\,\eps. \label{eq:weak-tail-bound}
\end{equation}
In particular, $c_\xi$ may be chosen so that $\beta_C(s)\le\eps/100$.
\end{lemma}

\begin{proof}
The first assertion follows from the definition of $S_{D,\xi}$ and Eq.~\eqref{eq:psd-cauchy}. 
If $s\notin\Delta_C$, no pair $u,v\in S_D$ with $u+v=s$ has both $u$ and $v$ in $C$. 
Each contributing pair therefore has at least one of $u$ and $v$ in $S_D\setminus C$ and satisfies $|A_{uv}|\le\sqrt\xi$. 
For a fixed $s$, the choice of $u$ determines $v=s+u$, so there are at most $|S_D|\le\sqrt{M_0}$ contributing ordered pairs. 
Summation gives Eq.~\eqref{eq:weak-tail-bound}.
\end{proof}

The conclusion is that every dissipative term at a label outside $\Delta_C$ contains at least one non-heavy row.  
Positive semidefiniteness then bounds the sum of all such terms by $O(\eps)$.

\subsection{Comparing the physical evolution with the no-jump evolution}

Define the true dissipative sum set as
\begin{equation}
\Delta_D:=\{u+v:u,v\in S_D\}\setminus\{0\}. 
\label{eq:DeltaD}
\end{equation}
Here, $u+v$ is the Pauli label obtained by multiplying $E_u$ and $E_v$, up to phase. 
The zero label may occur when $u=v$, but it is irrelevant for Hamiltonian support because there is no learned coefficient $h_0$.
The set $\Delta_D$ lists every label that can appear before cancellations are taken into account.
It may therefore be larger than the actual set of nonzero dissipative one-sided terms.

Let
\begin{align}
J:=\sum_{u,v\in S_D}A_{uv}E_vE_u,\qquad
K_0:=-iH-\frac{1}{2}J.
\end{align}
The operator $K_0$ describes the branch in which no jump occurs.  
The experiment never implements $e^{tK_0}$ as a separate physical evolution, and it appears only in the proof as a lower bound on the actual dynamics.  
Choose a factorization $A_{uv}=\sum_\alpha c_{\alpha u}\overline{c_{\alpha v}}$ and set $L_\alpha=\sum_u c_{\alpha u}E_u$.
With this notation, we have the following forms
\begin{align}
\Lop(\rho)=K_0\rho+\rho K_0^*+\Phi(\rho),\
\Phi(\rho)=\sum_\alpha L_\alpha\rho L_\alpha^*.
\end{align}
The map $\Phi$ is completely positive. 
The map $\rho\mapsto e^{tK_0}\rho e^{tK_0^*}$ is the $r=0$ completely positive term in the Dyson expansion of the full semigroup, which is a sum of completely positive contributions.  
The no-jump evolution is the term with no inserted jump map. 
Because all remaining terms are positive, every physical transition probability is at least its no-jump contribution.

\begin{lemma}[The physical evolution contains the no-jump contribution]
For every $t\ge0$ and every $\rho\succeq0$, we have
\begin{equation}
 e^{t\Lop}(\rho)\succeq e^{tK_0}\rho e^{tK_0^*}. \label{eq:no-jump-domination}
\end{equation}
\end{lemma}

\begin{proof}
Write $\Lop=\Lop_0+\Phi$ with $\Lop_0(X)=K_0X+XK_0^*$ and $\Phi(X)=\sum_\alpha L_\alpha X L_\alpha^*$.
The map $\Phi$ is completely positive, and for every $s\ge0$, the semigroup generated by $\Lop_0$ is
\begin{align}
 e^{s\Lop_0}(X)=e^{sK_0}Xe^{sK_0^*},
\end{align}
so it is also completely positive. 
Assuming finite dimension, the Dyson expansion is given by
\begin{align}
\begin{aligned}
e^{t(\Lop_0+\Phi)}&=\sum_{r=0}^\infty
\int_{0\le t_1\le\cdots\le t_r\le t}
e^{(t-t_r)\Lop_0}\Phi\cdot \\
&\quad e^{(t_r-t_{r-1})\Lop_0}\cdots\Phi e^{t_1\Lop_0}\,dt_1\cdots dt_r,
\end{aligned}
\end{align}
where the $r=0$ term is interpreted as $e^{t\Lop_0}$. 
Each integrand is a composition of completely positive maps, hence is completely positive. 
Therefore, for every $\rho\succeq0$, every term in the series is positive semidefinite when applied to $\rho$. 
The full output is the sum of the $r=0$ term, which is the no-jump contribution, and the positive semidefinite sum of all $r\ge1$ contributions $e^{t\Lop}(\rho)-e^{t\Lop_0}(\rho)\succeq0$.
Substituting the explicit form of $e^{t\Lop_0}$ gives the claimed result.
\end{proof}

The protocol uses the fact that their positive contribution is contained in the measured transition probability.
We also need to control how quickly the no-jump amplitudes change.  
The next lemma proves both the required contraction and the derivative bound.

\begin{lemma}[No-jump amplitudes contract and have bounded derivatives]\label{lem:no-jump-derivative-bound}
Let $K=K_0$ and let $B\in\{X,Z\}$ be a product-Pauli basis. For any displacement $d\in\mathbb F_2^n$, we define
\begin{align}
F_d^B(t)_a:=2^{-n/2}\langle a+d|e^{tK}|a\rangle_B.
\end{align}
Then, for every integer $\ell\ge0$ and every $t\ge0$, we have
\begin{align}
\|(F_d^B)^{(\ell)}(t)\|_2\le \Gamma^\ell.
\end{align}
\end{lemma}

\begin{proof}
First, we show that the no-jump evolution is a contraction. 
Using the factorization above, we have
\begin{align}
\sum_\alpha L_\alpha^*L_\alpha=\sum_{u,v}\left(\sum_\alpha \overline{c_{\alpha u}}c_{\alpha v}\right)E_uE_v.
\end{align}
After relabeling the dummy indices and using the convention in the definition of $J$, this simplifies to
\begin{align}
J=\sum_{u,v\in S_D}A_{uv}E_vE_u\succeq0.
\end{align}
Therefore, we have $K+K^*=-J\preceq0$.
Let $A(t)=e^{tK}$. Then
\begin{align}
\frac{d}{dt}\bigl(A(t)^*A(t)\bigr)=A(t)^*(K^*+K)A(t)\preceq0.
\end{align}
Since $A(0)^*A(0)=I$, integration in $t$ gives
\begin{align}
A(t)^*A(t)\preceq I,\qquad \|e^{tK}\|_\infty\le1.
\end{align}
Next, we note that
\begin{align}
\begin{aligned}
\|K\|_\infty&\le \|H\|_\infty+\frac{1}{2}\|J\|_\infty \\
&\le \|H\|_\infty+\frac{1}{2}\sum_{u,v\in S_D}|A_{uv}|\,\|E_vE_u\|_\infty \\
&\le \|H\|_\infty+\frac{1}{2}\sum_{u,v\in S_D}|A_{uv}|\le \Gamma.
\end{aligned}
\end{align}
We fix a product basis $B$ and a displacement $d$, and differentiate the amplitude vector to obtain
\begin{align}
(F_d^B)^{(\ell)}(t)_a=2^{-n/2}\langle a+d|K^\ell e^{tK}|a\rangle_B.
\end{align}
For any operator $M$, we have
\begin{align}
\begin{aligned}
\|F_d^B(M)\|_2^2&:=2^{-n}\sum_a |\langle a+d|M|a\rangle_B|^2 \\
&\le 2^{-n}\sum_a \|M|a\rangle_B\|_2^2 \\
&=2^{-n}\Tr(M^*M)\le \|M\|_\infty^2,
\end{aligned}
\end{align}
because $\Tr(M^*M)\le 2^n\|M\|_\infty^2$. 
Taking $M=K^\ell e^{tK}$ and using the contraction just proved, we have
\begin{align}
\begin{split}
\|(F_d^B)^{(\ell)}(t)\|_2&\le \|K^\ell e^{tK}\|_\infty\le \|K\|_\infty^\ell\|e^{tK}\|_\infty\\
&\le \Gamma^\ell,
\end{split}
\end{align}
which finishes the proof.
\end{proof}

\subsection{Hamiltonian projection strength outside the collision set}

We consider a Pauli string $P_s$ which has a symplectic form $(x(s)\mid z(s))$
In the $Z$ basis, the Pauli string $E_s$ changes the basis label by $x(s)$.  
In the $X$ basis, it changes the label by $z(s)$.  
Outside the collision set, the remaining dissipative term at label $s$ is small.  
A Hamiltonian coefficient larger than $\eps$ therefore creates a detectable derivative in at least one of these displacement amplitudes.

\begin{lemma}[A large Hamiltonian coefficient gives a detectable displacement derivative] \label{lem:projection-strength}
Assume $S_{D,\xi}\subseteq C$, let $s\notin\Delta_C$, and suppose $|h_s|>\eps$. 
If $x(s)=d\ne0$, we define
\begin{align}
F_d^Z(t)_a:=2^{-n/2}\langle a+d|e^{tK_0}|a\rangle_Z.
\end{align}
Then, we have
\begin{align}
\|(F_d^Z)'(0)\|_2\ge \eps/2.
\end{align}
Similarly, if $z(s)=e\ne0$, the analogous $X$-basis vector satisfies
\begin{align}
\|(F_e^X)'(0)\|_2\ge \eps/2.
\end{align}
\end{lemma}

\begin{proof}
We prove the $Z$-basis statement, and the $X$-basis statement follows a similar argument. 
Expanding $K_0=\sum_{x,z}g_{x,z}E_{(x|z)}$, for a computational-basis vector, we have
\begin{align}
E_{(x|z)}|a\rangle_Z=\theta_Z(x,z)(-1)^{z\cdot a}|a+x\rangle_Z
\end{align}
with $|\theta_Z(x,z)|=1$.
Hence, we obtain that
\begin{align}
(F_d^Z)'(0)_a=2^{-n/2}\sum _z g_{d,z}\theta_Z(d,z)(-1)^{z\cdot a}.
\end{align}
Walsh orthogonality then gives
\begin{equation}
\|(F_d^Z)'(0)\|_2^2=\sum_z|g_{d,z}|^2. 
\label{eq:projection-parseval}
\end{equation} 
Because $A=A^\dagger\succeq0$,
\begin{align}
J=\sum_{u,v\in S_D}A_{uv}E_vE_u=\sum_\alpha L_\alpha^*L_\alpha
\end{align}
is Hermitian and positive semidefinite.  
Its Pauli coefficient $j_s=2^{-n}\Tr(E_sJ)$ is therefore real.  
The coefficient of $E_s$ in $K_0=-iH-J/2$ is $g_s=-ih_s-j_s/2$.  
Since $h_s$ is also real, the Hamiltonian and dissipative parts lie in orthogonal imaginary and real directions:
\begin{align}
|g_s|^2=h_s^2+\frac14j_s^2\ge h_s^2.
\end{align}
The $s$-term in Eq.~\eqref{eq:projection-parseval} gives $\|(F_d^Z)'(0)\|_2\ge|h_s|>\eps$, which proves the stated bound. 
\end{proof}

\subsection{Projection from same-Pauli measurements}

The displacement experiment directly records the change of the basis label.  
A $Z$-basis shot draws $U\sim\pi_R$, sets $t=\tau U$, prepares a uniformly random computational-basis state $|a\rangle_Z$, evolves under $e^{t\Lop}$, measures again in the $Z$ basis, and records the displacement $d=a+b$. 
The $X$-basis experiment is defined in the same way.
Repeating these shots finds the $x$- and $z$-projections of every Hamiltonian label outside $\Delta_C$ whose coefficient is larger than $\eps$. 
This stage forms a candidate list as the result of this Lindbladian structure learning state, and the coefficient values are estimated in the next stage.

\begin{theorem}[Same-Pauli measurements find all large Hamiltonian projections outside the collision set] \label{thm:projection-stage}
Assume $S_{D,\xi}\subseteq C$. 
There is a control-free procedure with Pauli-product inputs and outputs using $\widetilde O(\Gamma^2/\eps^2)$ queries and total evolution time $\widetilde O(\Gamma/\eps^2)$ that outputs $\widehat\Pi_x,\widehat\Pi_z\subseteq\mathbb F_2^n$ such that, with probability at least $1-\delta$, every label $s\notin\Delta_C$ with $|h_s|>\eps$ satisfies
\begin{align}
x(s)\in\widehat\Pi_x\cup\{0\},\qquad
z(s)\in\widehat\Pi_z\cup\{0\}.
\end{align}
Each output list has size at most $\widetilde O(\Gamma^2/\eps^2)$, and the typical query time resolution is thus $O(\Gamma^{-1})$
\end{theorem}

\begin{proof}
Fix $s\notin\Delta_C$ with $|h_s|>\eps$ and $d=x(s)\ne0$. 
Conditional on the initial string $a$ and time $t$, no-jump domination gives
\begin{align}
\begin{aligned}
\hspace{-3.7mm}\Tr(|a+d\rangle\langle a+d|_Ze^{t\Lop}|a\rangle\langle a|_Z)\ge|\langle a+d|e^{tK_0}|a\rangle_Z|^2.
\end{aligned} 
\end{align}
Averaging over $a$ yields
\begin{align}
\Pr[D=d\mid t]\ge\|F_d^Z(t)\|_2^2.
\end{align}
\Cref{lem:no-jump-derivative-bound} supplies $\|(F_d^Z)^{(m)}(t)\|_2\le\Gamma^m$, and \Cref{lem:projection-strength} gives $\|(F_d^Z)'(0)\|_2\ge\eps/2$. 
Applying \Cref{lem:vector-observability} with parameter $\eps/2$ gives
\begin{align}
\Pr[D=d]\ge \frac{c\eps^2}{\Gamma^2R^6} =:p_*,\ R=O\!\left(\log(e+\Gamma/\eps)\right).
\end{align}
There are at most $M_H\le M_0$ target labels. 
Assuming $M_0\geq 1$, $C p_*^{-1}\log(2(M_0)/\delta)$ independent $Z$-basis shots contain every required nonzero $x$-displacement with failure probability at most $\delta/2$. 
The same argument in the $X$ basis handles the $z$-displacements. 
The claimed complexity and output-size bounds then follow.
\end{proof}

Conditioned on the projection to succeed, we define
\begin{align}
\widehat U_2:=\bigl((\widehat\Pi_x\cup\{0\})\times(\widehat\Pi_z\cup\{0\})\bigr)\setminus\{0\},\ 
\widehat U_3:=\widehat U_2\setminus\Delta_C.
\end{align}
Then, we have
\begin{equation}
\{s\notin\Delta_C:|h_s|>\eps\}\subseteq\widehat U_3. \label{eq:strong-H-in-U3}
\end{equation}
The structure-learning stage therefore produces the three classical objects
\begin{align}
 C,\quad \Delta_C,\quad \widehat U_3,
\end{align}
where the set $C$ localizes the dissipator, the set $\Delta_C$ marks the labels where Hamiltonian and dissipative one-sided effects may overlap, and the set $\widehat U_3$ contains every Hamiltonian label outside $\Delta_C$ whose coefficient is larger than $\eps$.


\section{Lindbladian coefficient learning}
\label{sec:lindbladian-coefficient-learning}

The structure-learning stage has reduced the problem to estimating the Lindbladian coefficients within the finite candidate sets.  
Outside $\Delta_C$, a cross-Pauli measurement isolates the Hamiltonian coefficient up to a small dissipative bias. 
Inside $\Delta_C$, the one-sided Hamiltonian and dissipative terms share the same labels, so we separate them through Pauli-Liouville inversion.
The approaches are summarized in the following table.

\begin{center}
\begin{tabular}{c|c|c}
coefficient & where & method\\ \hline
$h_s$ & $s\notin\Delta_C$ & linear response on $\widehat U_3$\\
$h_s$ & $s\in\Delta_C$ & Pauli-Liouville cancellation\\
$A_{uv}$ & $u,v\in C$ & fixed-sum Pauli-Liouville inversion
\end{tabular}
\end{center}

\subsection{Linear response outside the collision set}

The projection stage may return a strict superset of the Hamiltonian support.
We apply linear response only outside $\Delta_C$, where \Cref{lem:weak-tail} controls the remaining
dissipative contribution.

\paragraph{Frames and visible labels.}To estimate many Hamiltonian candidates with the same shots, we randomize the local preparation and measurement basis. 
We call the frame as a Pauli string $c=(c_1,\ldots,c_n)\in\{X,Y,Z\}^n$.
For every site choose the cyclic order $(c_i,a_i,b_i)$ among $(X,Y,Z)$, $(Y,Z,X)$, and $(Z,X,Y)$.
For a label $s$, define two labels $q_{c,s}$ and $o_{c,s}$ site by site as follows to ensure that $q_{c,s}+o_{c,s}=s$ in the symplectic form:
\begin{align}
\begin{array}{c|c|c}
\text{symplectic }s_i & \text{symplectic }(q_{c,s})_i & \text{symplectic }(o_{c,s})_i\\
\hline
I & I & I\\
a_i & a_i & I\\
b_i & I & b_i\\
c_i & a_i & b_i
\end{array}
\end{align}
Then $E_{q_{c,s}}$ is the product of the $a_i$ axes on the sites where $q_{c,s}$ is nonidentity, and $E_{o_{c,s}}$ is the product of the $b_i$ axes on the sites where $o_{c,s}$ is nonidentity. 
A label $s$ is visible in frame $c$ if the number of sites with local $c_i$ is odd. 
With the sitewise definition above, this is exactly the condition $[s,q_{c,s}]=1$, so the Hamiltonian commutator with $E_{q_{c,s}}$ is nonzero. 
When $s$ is not visible, we set $\sigma_c(s)=1$ by convention.

A label is useful in a frame only when its commutator with the prepared Pauli is nonzero, which is called visibility throughout this work.  
For a uniformly random frame, visibility is the parity of independent local matches, so every nonzero label is visible with constant probability as Ref.~\cite[Proposition C.6]{zhou2026optimal}.

\begin{lemma}[A random frame sees every nonzero Pauli label with constant probability, Proposition 6, Ref.~\cite{zhou2026optimal}]
For a uniformly random frame $C\in\{X,Y,Z\}^n$ and any nonzero Pauli label $s$,
\begin{align}
q(s):=\Pr[s\text{ visible in }C]=\frac{1-(1/3)^{\wt(s)}}{2}\ge \frac{1}{3}.
\end{align}
\end{lemma}

\paragraph{Dissipative bias outside the collision set.}The linear response contains both the desired Hamiltonian term and a dissipative term.  
Outside $\Delta_C$, every dissipative contribution with the same input-output label contains a non-heavy row.
The PSD row bound therefore makes the total bias small.

\begin{lemma}[The dissipative trace bias is small outside the collision set]\label{lem:dissipative-trace-bound}
Let $s\ne0$ and $q,o\in\V$ satisfy $o+q=s$. 
Then we have
\begin{equation}
\left|2^{-n}\Tr\bigl(E_o\Dop(E_q)\bigr)\right|\le 2\beta_C(s). 
\label{eq:dissipative-trace-bound}
\end{equation} 
Consequently, if $s\notin\Delta_C$ and $S_{D,\xi}\subseteq C$, the right-hand side is at most $\eps/50$.
\end{lemma}

\begin{proof}
Pauli orthogonality eliminates every summand whose Pauli label differs from $o$. 
For the jump term $E_uE_qE_v$ and both anticommutator terms, a nonzero trace against $E_o$ requires $u+v=s$.
Each normalized Pauli trace has modulus at most one. 
The absolute contribution of the summand indexed by $(u,v)$ is therefore at most
\begin{align}
|A_{uv}|\left(1+\frac{1}{2}+\frac{1}{2}\right) = 2|A_{uv}|.
\end{align}
Summing over $u+v=s$ gives Eq.~\eqref{eq:dissipative-trace-bound}, and then \Cref{lem:weak-tail} gives the
final assertion.
\end{proof}

\begin{lemma}[A visible frame isolates one Hamiltonian coefficient up to small bias]\label{lem:linear-response-bias}
Let $s\ne0$ be visible in a frame $c$. Then
\begin{equation}
2^{-n}\Tr\bigl(E_{o_{c,s}}\Lop(E_{q_{c,s}})\bigr)=2\sigma_c(s)h_s +b_{c,s}, 
\label{eq:linear-response}
\end{equation}
where we have $\sigma_c(s)\in\{\pm1\}$ and $|b_{c,s}|\le2\beta_C(s)$.
For $s\notin\Delta_C$ on the event $S_{D,\xi}\subseteq C$, one has $|b_{c,s}|\le\eps/50$.
\end{lemma}

\begin{proof}
Pauli orthogonality shows that the Hamiltonian trace vanishes for every label $r\ne s$. 
Visibility gives $[s,q_{c,s}]=1$, so $[E_s,E_{q_{c,s}}]=2E_sE_{q_{c,s}}$. 
Since $s+q_{c,s}=o_{c,s}$, there is a phase $\mu_c(s)\in\{\pm i\}$ such that $E_sE_{q_{c,s}}=\mu_c(s)E_{o_{c,s}}$. 
Thus, we have
\begin{align}
\begin{aligned}
&2^{-n}\Tr\left(E_{o_{c,s}}(-ih_s[E_s,E_{q_{c,s}}])\right)=-2i\mu_c(s)h_s\\
&\qquad=2\sigma_c(s)h_s,
\end{aligned}
\end{align}
where $\sigma_c(s):=-i\mu_c(s)\in\{\pm1\}$. 
The dissipative contribution is $b_{c,s}$, and \Cref{lem:dissipative-trace-bound} gives the claimed bound.
\end{proof}

\paragraph{Linear response estimators.}We note that one shot can be reused for many candidates. 
The shot prepares all $a_i$-axis eigenstates and measures all $b_i$ axes.  
Its local outcomes determine all parities $\prod_{i\in A}s_i\prod_{i\in B}m_i$ for all subsets $A,B$. 
Hence it serves every visible candidate $s$ at once. 
The samples for different $s$ inside the same shot are correlated, but concentration is applied for each fixed $s$ across independent shots, as the union bound does not require independence across different $s$.

For $s\in \mathcal U\subseteq\V\setminus(\Delta_C\cup\{0\})$, we define
\begin{align}
V_s:=\frac{\one\{s\text{ visible in }C\}}{2\sigma_C(s)q(s)}\frac{B_R}{\tau}\operatorname{sgn}(k_R(U))Z_s,
\end{align}
where $Z_s\in\{\pm1\}$ is the parity outcome for the trace experiment $E_{q_{C,s}}\to E_{o_{C,s}}$ at time $\tau U$.

For every visible candidate, the cross-Pauli trace response is linear in $h_s$ and has only a small bias bounded above. 
The same random-frame shot supplies a bounded parity sample for every visible candidate.  
We then apply concentration separately to each label.

\begin{theorem}[Random frames estimate all candidate Hamiltonian coefficients]\label{thm:linear-response}
Assume $S_{D,\xi}\subseteq C$. 
Given a finite $\mathcal U\subseteq\V\setminus(\Delta_C\cup\{0\})$ and $|\mathcal U|\geq 2$, the linear stage estimates every $h_s$, $s\in \mathcal U$, to error at most $\eps/4$ with probability at least $1-\delta$ using 
\begin{align}
\widetilde O\left(\frac{\Gamma^2}{\eps^2}\log \frac{|\mathcal U|}{\delta} \right)\quad\text{and}\quad
\widetilde O\!\left(\frac{\Gamma}{\eps^2}\log\frac{|\mathcal U|}{\delta}\right)
\end{align}
queries and total evolution time.
\end{theorem}

\begin{proof}
Fix $s\in \mathcal U$. 
For every frame $c$ in which $s$ is visible, let
\begin{align}
f_{c,s}(t):=2^{-n}\Tr\!\left(E_{o_{c,s}}e^{t\Lop}(E_{q_{c,s}})\right).
\end{align}
Lemma~\ref{lem:response-derivative-bound} supplies the derivative bounds required by \Cref{thm:one-sided-derivative}, while \Cref{lem:linear-response-bias} gives
\begin{align}
f'_{c,s}(0)=2\sigma_c(s)h_s +b_{c,s},\qquad |b_{c,s}|\le\eps/50.
\end{align}
Choose $R=C\lceil\log(e+\Gamma/\eps)\rceil$ so that $|D_Rf_{c,s} -f'_{c,s}(0) |\le \eps /24$.
Conditioning on the random frame and kernel time gives
\begin{align}
\begin{aligned}
\E[V_s]-h_s&=\frac{1}{q(s)}\E_C\!\left[\one\{s\text{ visible in }C\}\right.\\
&\qquad\left.\times\frac{b_{C,s}+D_Rf_{C,s}-f'_{C,s}(0)}{2\sigma_C(s)}\right].
\end{aligned}
\end{align}
Since $q(s)\ge1/3$, the visibility normalization cancels after averaging, and
\begin{align}
|\E[V_s]-h_s|\le\eps/100+\eps/48<\eps/24.
\end{align} 
Moreover, we have
\begin{align}
|V_s|\le \frac{B_R}{2q(s)\tau}\le C\Gamma R^3.
\end{align} 
Hoeffding's inequality and a union bound over $\mathcal U$ give sampling error at most $\eps/8$ for every
candidate when
\begin{align}
N\ge C\frac{\Gamma^2R^6}{\eps^2}\log\frac{|\mathcal U|}{\delta}.
\end{align}
The total error is smaller than $\eps/4$. 
Each shot has duration at most $\tau=O(1/\Gamma)$, yielding the stated total-time bound.
\end{proof}

\noindent Run \Cref{thm:linear-response} with $\mathcal U=\widehat U_3$ and denote the resulting estimates by
$\widehat h_s^{\rm lin}$.

\subsection{How Pauli-Liouville coefficients contain \texorpdfstring{$A_{uv}$}{Auv} and \texorpdfstring{$h_s$}{hs}}

The remaining coefficients can be separated after treating the generator itself as a linear map.  
We expand this map in the two-sided Pauli basis $X\mapsto E_aXE_b$.  
In this representation, $A_{uv}$ appears directly as one coordinate $R_{u,v}$.  
The Hamiltonian coefficient $h_s$ is the difference of two one-sided coordinates, and the dissipative anticommutator terms cancel in that difference.

We denote two-sided Pauli superoperators $T_{a,b}(X):=E_aXE_b$ for simplicity.
A Pauli-Liouville coefficient means the coefficient of a linear map in this superoperator basis. 
The maps $T_{a,b}$ form a basis of superoperators, so there are unique coefficients $R_{a,b}$ such that
\begin{equation}
\Lop(X)=\sum_{a,b\in\V}R_{a,b}E_aXE_b. 
\label{eq:PL-expansion}
\end{equation}
Comparing Eq.~\eqref{eq:model-L} with Eq.~\eqref{eq:PL-expansion} gives, for $u,v\ne0$, we have $A_{u,v}=R_{u,v}$.
To see the one-sided coefficients, expand the three parts of the generator in the basis $T_{a,b}(X)=E_aXE_b$. 
The jump term is already two-sided as $A_{uv}E_uXE_v=A_{uv}T_{u,v}(X)$.
The anticommutator term becomes one-sided after multiplying the two Pauli strings $E_vE_u=\omega(v,u)E_{u+v}$.
Thus a pair $u,v$ with $u+v=s$ contributes
\begin{align}
-\frac{1}{2} A_{uv}\omega(v,u)E_sX\quad\text{and}\quad-\frac{1}{2} A_{uv}\omega(v,u)XE_s
\end{align}
to $R_{s,0}$ and $R_{0,s}$. 
The Hamiltonian commutator contributes
\begin{align}
-ih_s[E_s,X]=-ih_sE_sX+ih_sXE_s.
\end{align}
Adding these contributions gives, for $s\ne0$, we have
\begin{align}
\begin{split}
&R_{s,0}= -ih_s-\frac{1}{2}\sum_{u+v=s}A_{uv}\omega(v,u),\\
&R_{0,s}= +ih_s-\frac{1}{2}\sum_{u+v=s}A_{uv}\omega(v,u),
\end{split}
\end{align}
with $h_s=\frac{R_{0,s}-R_{s,0}}{2i}$.

To estimate these coordinates with product-Pauli experiments, group them by the sum label $s=a+b$. 
For a fixed $s$, define the trace response
\begin{align}
g_s(q):=2^{-n}\Tr(E_{q+s}\Lop(E_q)).
\end{align}
For $a+b=s$, define the phase
\begin{align}
\gamma_{a,b}(q):=2^{-n}\Tr(E_{q+s}E_aE_qE_b)\in\{\pm1,\pm i\}.
\end{align}
All coefficients with $a+b=s$ contribute to the same family of trace responses $E_q\mapsto E_{q+s}$.  
Their phases form a Walsh system.  
Averaging against the matching phase therefore selects one coefficient from this fixed-sum block.

\begin{lemma}[Walsh averaging separates a fixed-sum Pauli-Liouville block]
For fixed $s$, we have
\begin{equation}
g_s(q)=\sum_{a+b=s}R_{a,b}\gamma_{a,b}(q). 
 \label{eq:sum-system}
\end{equation}
Moreover, we have
\begin{equation}
4^{-n}\sum_{q\in\V}
 \overline{\gamma_{a,a+s}(q)}\gamma_{a',a'+s}(q)
 =\one\{a=a'\}. \label{eq:gamma-orthog}
\end{equation}
Threfore, we obtain
\begin{equation}
R_{a,b}=4^{-n}\sum_{q\in\V}\overline{\gamma_{a,b}(q)}g_{a+b}(q). 
\label{eq:R-inversion}
\end{equation}
\end{lemma}

\begin{proof}
Substitute the Pauli-Liouville expansion Eq.~\eqref{eq:PL-expansion} into the definition of $g_s$:
\begin{align}
\begin{aligned}
g_s(q)&=2^{-n}\Tr\left(E_{q+s}\sum_{a,b}R_{a,b}E_aE_qE_b\right) \\
&=\sum_{a,b}R_{a,b}\,2^{-n}\Tr(E_{q+s}E_aE_qE_b).
\end{aligned}
\end{align}
The trace of a product of Pauli strings is zero unless the total Pauli label is zero. 
As the total label in the trace is $(q+s)+a+q+b=s+a+b$, only terms with $a+b=s$ survive, and the surviving normalized trace is precisely $\gamma_{a,b}(q)$. 
This proves Eq.~\eqref{eq:sum-system}.

It remains to prove orthogonality. 
Fix the sum label $s$ and write $\gamma_a(q):=\gamma_{a,a+s}(q)$.
Because $a+(a+s)=s$, the product inside the trace defining $\gamma_a(q)$ has total label zero, so $\gamma_a(q)$ is a unit phase. 
We isolate its dependence on $q$. 
Using the commutation character $\chi_a(q)=(-1)^{[a,q]}$, we may move $E_a$ past $E_q$ as $E_aE_q=\chi_a(q)E_qE_a$.
Therefore, we have
\begin{align}
E_aE_qE_{a+s}=\chi_a(q)E_q(E_aE_{a+s}).
\end{align}
The factor $E_aE_{a+s}$ is a fixed unit phase times $E_s$. We write $E_aE_{a+s}=c_{a,s}E_s$ with $|c_{a,s}|=1$.
Also as $E_qE_s=\omega(q,s)E_{q+s}$, we have
\begin{align}
\begin{aligned}
E_{q+s}E_aE_qE_{a+s}&=c_{a,s}\chi_a(q)E_{q+s}E_qE_s \\
&=c_{a,s}\chi_a(q)\omega(q,s)E_{q+s}^2 \\
&=c_{a,s}\chi_a(q)\omega(q,s)I.
\end{aligned}
\end{align}
Therefore
\begin{align} \gamma_a(q)=c_{a,s}\chi_a(q)\omega(q,s).
\end{align}
When we multiply one such phase by the conjugate of another with the same sum label $s$, the factor $\omega(q,s)$ cancels as
\begin{align}
\begin{aligned}
\overline{\gamma_a(q)}\gamma_{a'}(q)&=\overline{c_{a,s}}c_{a',s}\,\chi_a(q)\chi_{a'}(q)\\
&=c_{a,a',s}\,(-1)^{[q,a+a']},
\end{aligned}
\end{align}
where $c_{a,a',s}$ is a unit phase independent of $q$. 
If $a=a'$, then the left side is $|\gamma_a(q)|^2=1$ for every $q$. 
If $a\ne a'$, then $a+a'\ne0$, and the Walsh character $q\mapsto (-1)^{[q,a+a']}$ has zero average over $\V$. 
Hence, we prove the claimed Eq.~\eqref{eq:gamma-orthog}. 
Finally, multiply Eq.~\eqref{eq:sum-system} by $\overline{\gamma_{a,a+s}(q)}$, average over $q\in\V$, and use the orthogonality just proved. 
This selects the single coefficient $R_{a,a+s}$ and gives Eq.~\eqref{eq:R-inversion}.
\end{proof}

Let $N=4^n$ and define the matrix $(A_s)_{q,a}=\gamma_{a,a+s}(q)$. 
Then, we have $A_s^*A_s=NI$ and $A_s^{-1}=N^{-1}A_s^*$.
Every entry of $A_s^{-1}$ has modulus $1/N$, and each row has $N$ entries. 
Therefore, this indicates that $\|A_s^{-1}\|_{\infty\to\infty}=1$.

\subsection{Monte Carlo estimation of one Pauli-Liouville coefficient}

For fixed $(a,b)$ set $s=a+b$ and define
\begin{align}
F_{a,b}(t):=4^{-n}\sum_{q\in\V}\overline{\gamma_{a,b}(q)}2^{-n}\Tr(E_{q+s}e^{t\Lop}(E_q)).
\end{align}
Then by Eq.~\eqref{eq:R-inversion}, we have $F'_{a,b}(0)=R_{a,b}$.
Draw $q\sim\operatorname{Unif}(\V)$, run the trace-rule experiment $E_q\to E_{q+s}$, and multiply the outcome by $\overline{\gamma_{a,b}(q)}$. 
The sample remains bounded by one, and the derivative at zero is $R_{a,b}$.  
The one-sided derivative estimator therefore applies directly.

\begin{proposition}[Random trace measurements estimate one Pauli-Liouville coefficient]
For any $a,b\in\V$, one can estimate $R_{a,b}$ to additive error $\alpha$ with failure probability
$\delta$ using
\begin{align}
\widetilde O\left(\frac{\Gamma^2}{\alpha^2}\log\frac{1}{\delta}\right)\quad\text{and}\quad
\widetilde O\left(\frac{\Gamma}{\alpha^2}\log\frac{1}{\delta}\right)
\end{align}
queries and total time.
\end{proposition}

\begin{proof}
Fix $(a,b)$ and put $s=a+b$. 
For each $q\in\V$, the trace-rule experiment with preparation Pauli $E_q$ and measurement Pauli $E_{q+s}$ produces a bounded sample $Z_{q,t}$ satisfying
\begin{align}
\E[Z_{q,t}\mid q]=2^{-n}\Tr\left(E_{q+s}e^{t\Lop}(E_q)\right)
\end{align}
with $|Z_{q,t}|\le1$.
If $q$ is drawn uniformly from $\V$ and the sample is multiplied by $\overline{\gamma_{a,b}(q)}$, then the resulting random variable has expectation $F_{a,b}(t)$ and absolute value at most one. 
By Eq.~\eqref{eq:R-inversion} applied to the derivative of the channel at zero, we have
\begin{align}
\begin{aligned}
F'_{a,b}(0)&=4^{-n}\sum_{q\in\V}\overline{\gamma_{a,b}(q)}2^{-n}\Tr(E_{q+s}\Lop(E_q)\\
&=R_{a,b}.
\end{aligned}
\end{align}
The function $F_{a,b}$ is a normalized average of Pauli trace responses with coefficients $4^{-n}\overline{\gamma_{a,b}(q)}$. 
These coefficients have total absolute mass one. 
Therefore, \Cref{lem:response-derivative-bound} gives $|F_{a,b}^{(\ell)}(t)|\le \Gamma^\ell$ for every $\ell\ge0$. 
Applying \Cref{thm:one-sided-derivative} with target accuracy $\alpha$ estimates $F'_{a,b}(0)=R_{a,b}$ using
\begin{align}
\widetilde O\left(\frac{\Gamma^2}{\alpha^2}\log\frac{1}{\delta}\right)\quad\text{and}\quad
\widetilde O\left(\frac{\Gamma}{\alpha^2}\log\frac{1}{\delta}\right)
\end{align}
queries and total time.
\end{proof}

\subsection{Fixed-sum Pauli-Liouville reconstruction}
\label{subsec:overcomplete-reconstruction}

For one fixed sum label, the same random $q$ and the same trace-rule outcome can be multiplied by different known phases. 
Thus one data set serves every requested coefficient in that block.

\begin{proposition}[One fixed-sum data set estimates every requested coefficient in the block]\label{prop:fixed-sum-batch}
Fix $s\in\V$ and a finite set $I_s\subseteq\V$. 
A control-free product-Pauli procedure estimates all coefficients $\{R_{a,a+s}: a\in I_s\}$ to additive error at most $\alpha$ with failure probability at most $\delta_s$, using
\begin{align}
\widetilde O \! \left(\frac{\Gamma^2}{\alpha^2} \log \frac{|I_s|}{\delta_s}\right)\quad\text{and}\quad
\widetilde O\!\left(\frac{\Gamma}{\alpha^2}\log\frac{|I_s|}{\delta_s}\right)
\end{align}
queries and total evolution time assuming $|I_s|\geq 2$.
\end{proposition}

\begin{proof}
For $a\in I_s$, define
\begin{align}
\hspace{-2mm}F_{s,a}(t):=4^{-n}\sum_{q\in\V}\overline{\gamma_{a,a+s}(q)}2^{-n}\Tr(E_{q+s}e^{t\Lop}(E_q)).
\end{align}
Eq.~\eqref{eq:R-inversion} gives $F'_{s,a}(0)=R_{a,a+s}$. Draw $q\sim\operatorname{Unif}(\V)$ and perform the trace-rule experiment $E_q\to E_{q+s}$. 
If its bounded signed outcome is $Z_{q,t}$, the same shot supplies $\overline{\gamma_{a,a+s}(q)}Z_{q,t}$ for every $a\in I_s$. 
For each fixed $a$, this random variable has expectation $F_{s,a}(t)$ and modulus at most one. 
\Cref{thm:one-sided-derivative}, followed by a union bound over $I_s$, gives the simultaneous accuracy and the stated resources.
\end{proof}
 
For $s\in\Sigma_C$, define $I_s^D:=\{a\in C:a+s\in C\}$, and for $s\in\Delta_C$, set $I_s:=I_s^D\cup\{0,s\}$ and $I_0:=I_0^D$. 
Applying \Cref{prop:fixed-sum-batch} with $\alpha=\eps/4$ to every $s\in\Sigma_C$ estimates all $R_{u,v}=A_{uv}$ with $u,v\in C$ and both one-sided coefficients $R_{0,s},R_{s,0}$ for $s\in\Delta_C$. 
Since $|C|\le4M_0$, we have $|\Sigma_C|\le|C|^2\le16M_0^2$.
Consequently, the total evolution time and the number of shots for the Pauli-Liouville reconstruction are 
\begin{equation}
T_{\rm PL}=\widetilde O\!\left(\frac{M_0^2\Gamma}{\eps^2}\right),\qquad
N_{\rm PL}=\widetilde O\!\left(\frac{M_0^2\Gamma^2}{\eps^2}\right). \label{eq:batched-PL-time}
\end{equation}
For $s\in\Delta_C$, define
\begin{equation}
\widehat h_s^{\rm PL}:=\operatorname{Re}\!\left(\frac{\widehat R_{0,s}-\widehat R_{s,0}}{2i}\right). 
\label{eq:h-PL-overcomplete}
\end{equation}

\section{The complete reconstruction protocol and guarantee}

\subsection{The main guarantee}

We now combine the support and coefficient stages.
For completeness, we briefly recap the algorithm here in parallel with Algorithm~\ref{alg:lindblad} presented in the main text.
The input is $(M_0,\eps,\delta,\Gamma)$ with $\Gamma\ge\Gamma_*$.  
If $M_0=0$, the generator is zero and the protocol stops.  
Otherwise, we do the following steps
\begin{enumerate}[label=(\arabic*)]
\item Set $\xi=c_\xi\eps^2/M_0$ and construct $C$ by \Cref{thm:diss-row}.
\item Form $\Sigma_C$ and $\Delta_C$, run the projection stage, and construct $\widehat U_3$ as in Eq.~\eqref{eq:strong-H-in-U3}.
\item Run the linear estimator on $\widehat U_3$.
\item Run the fixed-sum estimator on every $s\in\Sigma_C$.
\item For $u,v\in C$, set
\begin{align}
\widehat A_{uv}:=\frac{1}{2}\left(\widehat R_{u,v}+\overline{\widehat R_{v,u}}\right),
\end{align}
and set $\widehat A_{uv}=0$ otherwise. Set
\begin{align}
\widehat h_s:=\begin{cases}
\widehat h_s^{\rm PL}, & s\in\Delta_C,\\
\widehat h_s^{\rm lin}, & s\in\widehat U_3,\\
0, & \text{otherwise}.
\end{cases}
\end{align}
\end{enumerate}

The theorem below guarantees the performance of the Lindbladian learning algorithm. 
Every coefficient that may exceed the target error is included in one of the estimated sets.  
Every omitted coefficient is already small by the PSD row bound or the projection guarantee.

\begin{theorem}[The overall performance guarantee]\label{thm:end-to-end}
Given an unknown Lindbladian in Eq.~\eqref{eq:model-L} with time evolution query access satisfying
\begin{align}
A=A^\dagger \succeq0,\qquad |h_s|\le1,\qquad |A_{uv}|\le1,
\end{align}
with $|S_H|+|S_D|^2\le M_0$ and known $\Gamma\ge\Gamma_*$. 
For every $0<\eps\le1$, the protocol above outputs estimates
satisfying
\begin{align}
\max_{s\ne0}|\widehat h_s-h_s|\le\eps,\qquad
\max_{u,v\ne0}|\widehat A_{uv}-A_{uv}|\le\eps
\end{align}
with probability at least $1-\delta$ using
\begin{equation}
N_\Sigma=\widetilde O\left(\frac{\Gamma^2M_0^2}{\eps^4}\right)\ \ \text{and}\ \ 
T_\Sigma=\widetilde O\left(\frac{M_0^2\Gamma}{\eps^2}\right)
\label{eq:main-time-eps}
\end{equation} 
shots and total evolution time.
For the safe scale $\Gamma=2M_0$, we have $T_\Sigma=\widetilde O(M_0^3/\eps^2)$
\end{theorem}
 
\begin{proof}
Allocate failure probability among candidate generation, projection, linear response, and fixed-sum
estimation. 
On their joint success event, \Cref{thm:diss-row} gives $S_{D,\xi}\subseteq C$ and $|C|\le4M_0$.

For $u,v\in C$, both $R_{u,v}$ and $R_{v,u}$ are estimated to error at most $\eps/4$.
Since $A_{vu}=\overline{A_{uv}}$, Hermitian symmetrization gives $|\widehat A_{uv}-A_{uv}|\le \tfrac12|\widehat R_{u,v}-A_{uv}|+\tfrac12|\widehat R_{v,u}-A_{vu}|\le\eps/4$.
If one of $u$ and $v$ lies outside $C$ and outside $S_D$, the true entry is zero.
If one of $u$ and $v$ lies in $S_D\setminus C$, \Cref{lem:weak-tail} gives $|A_{uv}|\le\sqrt\xi\le\eps/100$ after fixing $c_\xi$. 
Therefore, every Kossakowski entry has error at most $\eps$.

For $s\in\Delta_C$, the identity $h_s=\tfrac{R_{0,s}-R_{s,0}}{2i}$ and the two Pauli-Liouville estimates give $|\widehat h_s^{\rm PL}-h_s|\le\eps/4$. 
For $s\in\widehat U_3$, Theorem~ \ref{thm:linear-response} gives $|\widehat h_s^{\rm lin}-h_s|\le\eps/4$. 
Finally, if $s\notin\Delta_C\cup\widehat U_3$, inclusion Eq.~\eqref{eq:strong-H-in-U3} implies $|h_s|\le\eps$, and the protocol outputs zero. 
Thus the Hamiltonian error is at most $\eps$.

\Cref{thm:projection-stage} contributes $\widetilde O(\Gamma^2M_0^2/\eps^4)$ shots and $\widetilde O(M_0/\eps^2)$ times.
\Cref{thm:linear-response} contribute
$\widetilde O(\Gamma^2/\eps^2)$ shots and
$\widetilde O(\Gamma/\eps^2)$ time. 
Eq.~\eqref{eq:batched-PL-time} gives the fixed-sum cost.
For $M_0\ge1$ and $\eps\le1$, summing these terms yields Eq.~\eqref{eq:main-time-eps}.
\end{proof}

\subsection{Discretized Lindbladian learning protocol}

The previous estimators draw a time from a continuous distribution. 
Similar to Ref.~\cite{zhou2026optimal}, we can discretize the protocol.
However, the approach in Ref.~\cite{zhou2026optimal} cannot be directly applied as the Lindbladian does not have a bounded bandwidth as a Hamiltonian.
We can replace the distribution by finitely many positive times and then snap those times to a hardware lattice.  
Neither step changes the stated query or total-evolution-time scaling.

\paragraph{Replacing the continuous time draw by finitely many positive times.}A finite interpolation rule can reproduce the derivative estimate. 
We use Chebyshev-Lobatto nodes~\cite{trefethen2019approximation} because they give stable endpoint differentiation.  
The value at $t=0$ is known from Pauli algebra, so it is inserted classically and never counted as a channel use.

Set $\tau:=\tfrac{1}{2\Gamma}$, for an integer $R\ge2$, define the Chebyshev-Lobatto nodes on $[0,\tau]$ by
\begin{align}
x_j=\cos\frac{j\pi}{R},\quad
t_j=\frac{\tau}{2}(1-x_j),\quad j=0,1,\ldots,R.
\end{align}
Thus, we have $0=t_0<t_1<\cdots<t_R=\tau$ with $t_1\ge \tfrac{c}{\Gamma R^2}$ for a universal constant $c>0$. 
Let $\ell_j$ be the Lagrange basis polynomial satisfying $\ell_j(t_m)=\mathbf 1\{j=m\}$, and define endpoint derivative weights $w_j:=\ell'_j(0)$ and $W_R:=\sum_{j=0}^R |w_j|$.
Only a logarithmic number of nodes is needed.  
The next lemma bounds the interpolation weights and the resulting derivative bias.

\begin{lemma}[A finite set of positive times estimates the derivative at zero]\label{lem:finite-node-estimator}
There is a universal constant $C$ such that $W_R\le C\Gamma R^3$.
Moreover, if a scalar response $f$ satisfies
\begin{align}
|f^{(m)}(t)|\le \Gamma^m,\quad 0\le t\le \tau,\quad m=0,1,2,\ldots,
\end{align}
then, we have
\begin{align}
\left|\sum_{j=0}^R w_j f(t_j)-f'(0)\right|\le C\Gamma R^3 2^{-R}. 
\label{eq:finite-node-bias}
\end{align}
Consequently, for target derivative accuracy $\alpha$, choosing $R=C\left\lceil \log(e+\Gamma/\alpha)\right\rceil$ gives bias at most $\alpha/4$.
\end{lemma}

\begin{proof}
We compute the first positive node. 
Since $t_1=\tfrac{\tau}{2}\left(1-\cos\frac{\pi}{R}\right)$ and $1-\cos x\ge c x^2$ for $0\le x\le\pi/2$, we have
\begin{align}
t_1\ge c\tau R^{-2}=\frac{c}{\Gamma R^2}\end{align}
after adjusting the universal constant.

We next bound the total variation of the derivative weights. 
For any data vector $y=(y_0,\ldots,y_R)$ with $|y_j|\le1$, let $p_y$ be the degree-$R$ interpolating polynomial satisfying $p_y(t_j)=y_j$. 
Since
\begin{align}
p_y'(0)=\sum_{j=0}^R y_j\ell_j'(0)=\sum_{j=0}^R y_jw_j,
\end{align}
duality for the $\ell_1$ norm gives
\begin{align}
W_R=\sum_j|w_j|=\sup_{|y_j|\le1}|p_y'(0)|.
\end{align}
Chebyshev-Lobatto interpolation has Lebesgue constant $O(\log R)$, so
\begin{align}
\|p_y\|_{L^\infty[0,\tau]}\le C\log R\max_j|y_j|\le C\log R.
\end{align}
Markov's inequality on an interval of length $\tau$ states that a degree-$R$ polynomial obeys
\begin{align}
\|p_y'\|_{L^\infty[0,\tau]}\le \frac{CR^2}{\tau}\|p_y\|_{L^\infty[0,\tau]}.
\end{align}
Since $\tau^{-1}=2\Gamma$, this gives
\begin{align}
|p_y'(0)|\le C\Gamma R^2\log R\le C\Gamma R^3.
\end{align}
Taking the supremum over $y$ proves $W_R\le C\Gamma R^3$.

Now let $T_Rf$ be the Taylor polynomial of $f$ at zero of degree $R$. 
The differentiation rule $\sum_jw_jh(t_j)$ for every polynomial $h$ of degree at most $R$ indicates that $\sum_jw_j(T_Rf)(t_j)=(T_Rf)'(0)=f'(0)$.
For $0\le t\le\tau$, Taylor's theorem and $\Gamma\tau=1/2$ give
\begin{align}
|f(t)-T_Rf(t)| \le \frac{\Gamma^{R+1}t^{R+1}}{(R+1)!}\le \frac{2^{-(R+1)}}{(R+1)!}\le 2^{-R}.
\end{align}
Therefore, we have
\begin{align}
\big|\sum_jw_jf(t_j)-f'(0)\big|&=\big|\sum_jw_j\bigl(f(t_j)-(T_Rf)(t_j)\bigr)\big| \nonumber\\
&\le W_R\max_j|f(t_j)-(T_Rf)(t_j)| \nonumber\\
&\le C\Gamma R^3 2^{-R}.
\end{align}
Choosing $R=C\lceil\log(e+\Gamma/\alpha)\rceil$ with a sufficiently large universal constant makes this last expression at most $\alpha/4$.
\end{proof}

\begin{lemma}[A finite set of positive times detects a large vector derivative]\label{lem:finite-node-vector-observability}
Let $F:[0,\tau]\to\mathcal H$ be a finite-dimensional Hilbert-space-valued function satisfying $\|F^{(m)}(t)\|\le \Gamma^m$ for $0\le t\le\tau$ and $m=0,1,2,\ldots$.
If $\|F'(0)\|\ge\eta$ and $R=C\lceil\log(e+\Gamma/\eta)\rceil$, then
\begin{align}
\E_{J}[\|F(t_J)\|^2]=\sum_{j=0}^R\frac{|w_j|}{W_R}\|F(t_j)\|^2\ge\frac{c\eta^2}{\Gamma^2R^6}.
\end{align}
\end{lemma}

\begin{proof}
Let $T_RF$ be the vector-valued Taylor polynomial of degree $R$ at zero. 
The scalar interpolation identity applies componentwise, so $\sum_j w_j(T_RF)(t_j)=(T_RF)'(0)=F'(0)$.
Taylor's theorem in norm gives, for $0\le t\le\tau$,
\begin{align}
\|F(t)-T_RF(t)\|\le \frac{\Gamma^{R+1}t^{R+1}}{(R+1)!}\le 2^{-R}.
\end{align}
Hence, using \Cref{lem:finite-node-estimator},
\begin{align}
&\left\|\sum_jw_jF(t_j)-F'(0)\right\|\le \sum_j|w_j|\,\|F(t_j)-T_RF(t_j)\| \nonumber\\
&\qquad\qquad\le W_R2^{-R}\le C\Gamma R^3 2^{-R}\le \eta/2
\end{align}
for the stated choice of $R$. 
Therefore, we have $\left\|\sum_jw_jF(t_j)\right\|\ge\eta/2$.
On the other hand, with $J$ drawn with probability $|w_J|/W_R$, we have
\begin{align}
\begin{aligned}
 \left\|\sum_jw_jF(t_j)\right\|^2
 &=W_R^2\left\|\E_J\big[\operatorname{sgn}(w_J)F(t_J)\big]\right\|^2 \\
 &\le W_R^2\E_J\|F(t_J)\|^2,
\end{aligned}
\end{align}
by Jensen's inequality. 
Since $W_R\le C\Gamma R^3$, the claimed lower bound follows.
\end{proof}

The finite-node estimator is implemented without treating the zero node as a positive-time use of the channel. 
Split the interpolation formula as
\begin{align}
 \sum_{j=0}^Rw_jf(t_j)=w_0f(0)+\sum_{j=1}^Rw_jf(t_j).
\end{align}
The first term is inserted deterministically from Pauli algebra. 
For example, in a trace response $f_{P,Q}$, one has $f_{P,Q}(0)=2^{-n}\operatorname{Tr}(PQ)$, for a diagonal Pauli-$\chi$ response
$f_u(t)=\chi_{u,u}(t)$, one has $f_u(0)=\mathbf 1\{u=0\}$, and for a normalized Pauli-Liouville average, the value at $t=0$ is obtained by the same known linear combination of identity-map trace responses. 
Let $W_R^+:=\sum_{j=1}^R |w_j|\le W_R$.
When $W_R^+>0$, draw a positive node $J\in\{1,\ldots,R\}$ with probability $|w_J|/W_R^+$, run the trace-rule experiment at the positive time $t_J$, and output
\begin{align}
W=w_0f(0)+W_R^+\operatorname{sgn}(w_J)Z_J
\end{align}
at $\mathbb E[Z_J]=f(t_J)$ and $|Z_J|\le1$.
Then $\mathbb E[W]=\sum_jw_jf(t_j)$ and $|W|\le |w_0|+W_R^+ = W_R$.
The same Hoeffding argument therefore gives the shot and total evolution time
\begin{align}
N=\widetilde O\left(\frac{\Gamma^2}{\alpha^2}\right), \qquad
T=\widetilde O\left(\frac{\Gamma}{\alpha^2}\right),
\end{align}
where $N$ counts only positive-time uses of the channel. 
The smallest nonzero designed time is not
arbitrarily close to zero; it obeys
\begin{align}
t_1\ge \frac{c}{\Gamma\log^2(e+\Gamma/\alpha)}. \label{eq:min-node-time}
\end{align}
For finite-node projection sampling, the vector-valued lemma may also be conditioned on $J\ge1$.
Indeed, the target displacement from a same-Pauli measurement is nonzero, so its amplitude vector satisfies $F(0)=0$.
Removing the zero node can only increase the average of $\|F(t_J)\|^2$ over the sampled positive nodes.

\Cref{thm:projection-stage} already uses the single positive time $\tau_D$ and requires no
zero-time channel query.
Replace the continuous draw $U\sim\pi_R$ by the positive finite-node draw $J\in\{1,\ldots,R\}$ with probability $|w_J|/W_R^+$. 
If $d$ is the nonzero target displacement associated with a label $s\notin\Delta_C$ satisfying $|h_s|>\eps$, \Cref{lem:no-jump-derivative-bound} and \Cref{lem:projection-strength} give
\begin{align}
\|(F_d^B)'(0)\|\ge \eps/2,\qquad
\|(F_d^B)^{(m)}(t)\|\le\Gamma^m.
\end{align}
Applying \Cref{lem:finite-node-vector-observability} with $\eta=\eps/2$ gives
\begin{align}
\sum_{j=0}^R\frac{|w_j|}{W_R}\|F_d^B(t_j)\|^2\ge\frac{c\eps^2}{\Gamma^2R^6}.
\end{align}
For a nonzero displacement, $F_d^B(0)=0$. 
Removing the zero node and renormalizing can only
increase the weighted average as
\begin{align}
\begin{split}
\sum_{j=1}^R\frac{|w_j|}{W_R^+}\|F_d^B(t_j)\|^2&=\frac{W_R}{W_R^+} \sum_{j=0}^R\frac{|w_j|}{W_R}\|F_d^B(t_j)\|^2\\
&\ge \frac{c\eps^2}{\Gamma^2R^6}.
\end{split}
\end{align}
By no-jump domination, the actual displacement probability in the same-Pauli measurement at a positive node $t_J$ is at least $\|F_d^B(t_J)\|^2$. 
Hence one positive finite-node projection shot hits the target displacement with probability at least $c\eps^2/(\Gamma^2R^6)$, the lower bound used in the covering argument. 
Thus every use of the continuous signed-kernel estimator and every vector-valued observability step in the main proof may be replaced by the finite-node construction without changing the stated asymptotic resources or using zero-time channel queries.

\paragraph{Rounding the designed times to a hardware lattice.}A laboratory may allow only integer multiples of a basic time step $\Delta_t$.  
We therefore round each designed node to a nearby lattice point.  
The only new error is the change in the response caused by this rounding, which is controlled by the first-derivative bound.

Suppose now that the available positive times lie on a lattice
\begin{align}
\Delta_t\mathbb N=\{\Delta_t,2\Delta_t,3\Delta_t,\ldots\}.
\end{align}
Choose lattice points $\widetilde t_j\in\Delta_t\mathbb N\cup\{0\}$ satisfying $|\widetilde t_j-t_j|\le\Delta_t/2$. 
If $\Delta_t\le c/(\Gamma R^2)$, the snapped nodes may be chosen distinct and ordered. The derivative bias caused by snapping is bounded directly from $|f'|\le\Gamma$:
\begin{align}
\begin{aligned}
\left|\sum_j w_j\big(f(\widetilde t_j)-f(t_j)\big)\right|&\le \sum_j |w_j|\,\Gamma |\widetilde t_j-t_j| \\
&\le C\Gamma^2R^3\Delta_t.
\end{aligned} 
\label{eq:snap-bias}
\end{align}
Therefore, to estimate a derivative to accuracy $\alpha$, it is sufficient to impose
\begin{equation}
\Delta_t\le \frac{c\alpha}{\Gamma^2\log^3(e+\Gamma/\alpha)}. 
\label{eq:lattice-sufficient}
\end{equation}
Under Eq.~\eqref{eq:lattice-sufficient}, the snapped finite-node estimator has the same query and total-time scaling as \Cref{lem:finite-node-estimator}, after changing only universal constants. 
The same snapping estimate holds in norm for the vector-valued projection functions, since
\begin{align}
\begin{aligned}
&\left\|\sum_jw_j\big(F(\widetilde t_j)-F(t_j)\big)\right\| \\
&\quad\le \sum_j|w_j|\,\Gamma |\widetilde t_j-t_j|\le C\Gamma^2R^3\Delta_t.
\end{aligned}
\end{align}
Thus the finite-node projection stage also survives the lattice snapping condition. 
The dissipator candidate stage additionally
requires $\Delta_t\le c\tau_D$.  
Since $\tau_D=c_0\tfrac{\xi}{\Gamma^2}=c_0c_\xi\frac{\eps^2}{M_0\Gamma^2}$, a sufficient end-to-end lattice condition is
\begin{equation}
\Delta_t\le c\min\left\{\frac{\eps^2}{M_0\Gamma^2},\frac{\eps}{\Gamma^2\log^3(e+\Gamma/\eps)}\right\}.
\label{eq:lattice-end-to-end}
\end{equation}

For the \Cref{thm:linear-response}, use the lattice time
\begin{align}
\widetilde{\tau}_D:=\Delta_t\left\lceil\frac{\tau_D}{\Delta_t}\right\rceil.
\end{align}
If $\Delta_t\le c\tau_D$, then
\begin{align}
\tau_D\le \widetilde{\tau}_D\le (1+c)\tau_D.
\end{align}
Replacing $\tau_D$ by $\widetilde{\tau}_D$, and $p_\xi=\tau_D\xi$ by $\widetilde p_\xi=\widetilde{\tau}_D\xi$, changes only universal constants in the candidate-generation bounds.

\section{SPAM robustness}
\label{sec:spam-robustness}

State-preparation-and-measurement (SPAM) noise changes the intended input states and measurements.  
We assume that this noise has been calibrated.  
It affects the two data types differently: same-Pauli measurements need reliable displacement records, while trace responses are multiplied by a known reliability factor.

A standard model is independent single-qubit depolarizing SPAM~\cite{romanov2026learning,ivashkov2026ansatz,zhou2026optimal}. 
For an intended single-qubit input state $\rho$, the actual prepared state is
\begin{align}
\rho \mapsto {\cal E}_P(\rho):=r_P\rho+(1-r_P)\frac{I}{2},
\end{align}
and immediately before measurement the state is acted on by
\begin{align}
\rho \mapsto {\cal E}_M(\rho):=r_M\rho+(1-r_M)\frac{I}{2}.
\end{align}
The retention factors $r_P$ and $r_M$ are known from calibration. 
Thus one experimental shot implements
\begin{align}
\rho \mapsto {\cal E}_M^{\otimes n}\circ e^{t\Lop}\circ {\cal E}_P^{\otimes n}(\rho),
\end{align}
followed by the intended terminal Pauli-basis measurement. 
We write $r:=r_Pr_M$ for the combined one-qubit retention factor.

The depolarizing model is a useful example, but the correction only needs a more general diagonal action on Pauli observables.  
We say that a product SPAM layer is calibrated Pauli diagonal when, on each qubit, the preparation and measurement noise ${\cal P}_i(\cdot)$ and ${\cal M}^{*}_i$ are known unital Pauli maps. 
In the Pauli basis,
\begin{align}
\begin{aligned}
{\cal P}_i(P)&=\alpha_i(P)P,\quad
{\cal M}^{*}_i(P)&=\beta_i(P)P,
\end{aligned}
\end{align}
with known real numbers $\alpha_i(P)$ and $\beta_i(P)$, and $P\in\{X,Y,Z\}$. 
Here, the second describes measurement noise after moving it onto the measured observable. 
For a prepared Pauli string $Q$ and a measured Pauli string $P$, set
\begin{align}
\alpha(Q):=\prod_{i:Q_i\ne I}\alpha_i(Q_i),\quad
\beta(P):=\prod_{i:P_i\ne I}\beta_i(P_i).
\end{align}
In the depolarizing model, we have
\begin{align}
\alpha(Q)=r_P^{\wt(Q)},\quad \beta(P)=r_M^{\wt(P)}.
\end{align}
The product $\alpha(Q)\beta(P)$ is the reliability factor for the trace experiment. 
It is the known factor by which SPAM multiplies the ideal trace response. 
We assume that every trace experiment used in the coefficient stages has a nonzero calibrated reliability factor, and define
\begin{equation}
\zeta_2:=\inf_{\text{trace experiments}(P,Q)} |\alpha(Q)\beta(P)|>0. \label{eq:trace-spam-visibility}
\end{equation}

\paragraph{Structure learning under SPAM.}The structure-learning stage uses three flip distributions from same-Pauli measurements to construct $C$.
Rather than deriving a full detector-specific SPAM model, we state the exact calibrated guarantee needed by the proof.

\begin{assumption}[Calibrated SPAM preserves the heavy-row candidate guarantee]\label{ass:spam-candidate}
At the threshold $\xi$ and time $\tau_D$, a calibrated implementation of the \Cref{thm:diss-row} outputs $C$ satisfying
\begin{align}
 S_{D,\xi}\subseteq C,
\qquad |C|\le C_0M_0
 \end{align}
with the same failure probability as \Cref{thm:diss-row}. 
Its shot and total evolution time costs are at
most a constant $K_1\ge1$ times the ideal costs.
\end{assumption}

The projection stage records the difference between the prepared and measured basis strings. 
Its proof only needs a lower bound on the probability that a true target displacement is recorded unchanged.  
We denote this calibrated probability by $\zeta_1$.

\begin{assumption}[Calibrated SPAM preserves each target displacement with probability $\zeta_1$]\label{ass:disp-spam}
In a $Z$- or $X$-basis projection shot, let $D_0$ be the ideal displacement that would have been recorded without SPAM, and let $D_1$ be the recorded displacement after SPAM.
There is a calibrated number $\zeta_1>0$ such that, uniformly over the basis, starting string, evolution time, and every nonzero displacement used in the projection proof,
\begin{align}
\Prb[D_1=d\mid D_0=d]\ge \zeta_1.
\end{align}
The SPAM randomness is independent of the Lindblad evolution and of the classical time draw.
\end{assumption}

\noindent Under Assumption~\ref{ass:disp-spam}, an ideal target-displacement probability $p_*$ becomes at least $\zeta_1p_*$ after SPAM. 
Multiplying the projection repetition count by $\zeta_1^{-1}$ preserves the covering guarantee.

\paragraph{Coefficient learning under SPAM.}Both coefficient-learning stages use the same primitive: prepare a Pauli eigenstate for $Q$, evolve for time $t$, and measure $P$. 
Under calibrated Pauli-diagonal SPAM, this trace response is only multiplied by a known preparation-and-measurement reliability factor.  
Dividing by that factor recovers an unbiased ideal response.

\begin{lemma}[Calibrated Pauli-diagonal SPAM only rescales a trace response]\label{lem:spam-trace-rule}
Under calibrated Pauli-diagonal SPAM, the cross-Pauli measurement targeting $(P,Q)$ satisfies
\begin{equation}
\E_S[rm]=\alpha(Q)\beta(P)\,2^{-n}\Tr\!\left(Pe^{t\Lop}(Q)\right). 
\label{eq:spam-trace-rule}
\end{equation}
Consequently, dividing the signed outcome by $\alpha(Q)\beta(P)$ gives an unbiased sample of the ideal trace response, with sample magnitude inflated by at most $\zeta_2^{-1}$.
\end{lemma}

\begin{proof}
Let $\{\rho_\psi\}$ be the ideal product eigenbasis used for the trace-rule experiment targeting $Q$, with eigenvalue signs $r(\psi)$. 
In the ideal experiment, we have $\sum_\psi r(\psi)\rho_\psi=Q$.
With preparation noise, the signed prepared operator becomes $\sum_\psi r(\psi) {\cal P}(\rho_\psi)={\cal P}(Q)=\alpha(Q)Q$. 
For the measurement noise, we move the noise map from the final state to the observable. 
Thus, measuring after the noisy measurement layer is equivalent to measuring ${\cal M}^{*}(P)=\beta(P)P$.
Therefore, the signed observed average is
\begin{align}
\begin{aligned}
\E_S[rm]&=2^{-n}\Tr\!\left({\cal M}^{*}(P)e^{t\Lop}({\cal P}(Q))\right) \\
&=2^{-n}\Tr\!\left(\beta(P)P\,e^{t\Lop}(\alpha(Q)Q)\right) \\
&=\alpha(Q)\beta(P)\,2^{-n}\Tr\!\left(Pe^{t\Lop}(Q)\right).
\end{aligned}
\end{align}
If $\alpha(Q)\beta(P)\ne0$, division by this known scalar gives an unbiased ideal-response
sample. 
The original outcome has magnitude at most one, so the corrected outcome has magnitude at most $|\alpha(Q)\beta(P)|^{-1}\le\zeta_2^{-1}$.
\end{proof}

In the coefficient learning stage, all random samples used in Hoeffding bounds are larger by at most $\zeta_2^{-1}$, so the corresponding sample counts are multiplied by at most $\zeta_2^{-2}$.
In the depolarizing special case, the correction factor for a trace response $(P,Q)$ is $r_P^{-\wt(Q)}r_M^{-\wt(P)}$.
If the relevant preparation and measurement Pauli strings both have weight at most $k$, then $|\alpha(Q)\beta(P)|\ge (r_Pr_M)^k=r^k$, and the worst-case variance overhead is at most $r^{-2k}$.
Wrapping all these blocks together, we obtain the following result.

\paragraph{The overall SPAM guarantee.}After the calibrated corrections above, the ideal expectations and support guarantees are unchanged.
Only the number of repetitions increases: candidate generation pays $K_1$, displacement detection
pays $\zeta_1^{-1}$, and trace-response concentration pays $\zeta_2^{-2}$.

\begin{corollary}[Under the calibrated SPAM assumptions, only repetition costs change]\label{prop:spam-resource}
Assume the hypotheses of \Cref{thm:end-to-end} , calibrated SPAM satisfying \Cref{ass:spam-candidate} and~\Cref{ass:disp-spam}, and a nonzero trace-response reliability factor as in \eqref{eq:trace-spam-visibility}. 
Use the calibrated candidate procedure, multiply projection repetitions by $\zeta_1^{-1}$, and divide every trace-response sample by its calibrated reliability factor.
Then the coefficient guarantee of \Cref{thm:end-to-end} holds with probability at least $1-\delta$, and shot number and total evolution time of
\begin{align}
\begin{aligned}
N_\Sigma^S&=\widetilde O\!\left(K_1\frac{\Gamma^2M_0^2}{\eps^4}+\zeta_1^{-1}\frac{\Gamma^2}{\eps^2}+\zeta_2^{-2}\frac{M_0^2\Gamma^2}{\eps^2}\right),\\
T_\Sigma^S&=\widetilde O\!\left(K_1\frac{M_0}{\eps^2}+\zeta_1^{-1}\frac{\Gamma}{\eps^2}+\zeta_2^{-2}\frac{M_0^2\Gamma}{\eps^2}\right).
\end{aligned}
\end{align}
\end{corollary}

\end{document}